\documentclass{article}

    \PassOptionsToPackage{numbers, compress}{natbib}

\usepackage[preprint]{neurips_2026}

\usepackage[utf8]{inputenc} 
\usepackage[T1]{fontenc}    
\usepackage{hyperref}       
\usepackage{url}            
\usepackage{graphicx}
\usepackage{booktabs}       
\usepackage{amsmath,amssymb,amsfonts,amsthm} 
\newtheorem{proposition}{Proposition}
\newtheorem{corollary}[proposition]{Corollary}
\newtheorem{remark}[proposition]{Remark}
\usepackage{nicefrac}       
\usepackage{microtype}      
\usepackage{xcolor}         
\usepackage{colortbl}
\usepackage{enumitem}
\usepackage{subcaption}
\usepackage{multirow}
\usepackage[misc]{ifsym}

\definecolor{darkblue}{rgb}{0, 0, 0.5}
\hypersetup{colorlinks=true, citecolor=darkblue, linkcolor=darkblue, urlcolor=darkblue}

\title{MAST: Label-Efficient, Robust, and Generalizable Sound Detection for Biodiversity Monitoring via Masked Audio Pretraining and Self-Training}

\author{
\textbf{Tianyi Xu}$^{1}$ \quad
\textbf{Daniel L. Pimentel-Alarc\'{o}n}$^{1}$ \\
\textbf{Zuzana Bu\v{r}ivalov\'{a}}$^{1}$ \quad
\textbf{Claudia Sol\'{i}s-Lemus}$^{1,*}$ \\
\vspace{2pt}
$^{1}$\textnormal{University of Wisconsin--Madison}
}

\begin{document}

\maketitle

\begingroup
\renewcommand{\thefootnote}{*}
\footnotetext{Corresponding author: Claudia Sol\'{i}s-Lemus. Correspondence: \texttt{txu223@wisc.edu}, \texttt{solislemus@wisc.edu}.}
\endgroup

\begin{abstract}
Passive acoustic monitoring can measure biodiversity at larger scales, but time--frequency annotation of animal vocalizations is expensive, site-specific, and difficult to sustain at scale. We present a label-efficient sound detection framework that combines masked audio pretraining with a lightweight detector on mel spectrograms, then further improves robustness through iterative self-training on unlabeled audio. We first pretrain a ViT-based encoder on unlabeled recordings via masked reconstruction and transfer the encoder to a detection backbone. To better separate animal sounds from confounding background, we add a box-level contrastive loss that pulls matched event regions together while pushing noisy negatives apart. We then apply a two-stage pseudo-labeling curriculum to exploit large unlabeled pools without additional annotation. We evaluate the performance on two ecologically distinct domains: tropical rainforest soundscapes (Indonesia) and bird vocalizations in Mediterranean habitats (Spain). On both domains, masked audio pretraining and contrastive learning consistently improve time--frequency detection under temporal and cross-site distribution shift, and self-training yields further gains in out-of-distribution performance. On the rainforest domain, MAST with self-training achieves +0.22 mAP and +0.24 F1 over the strongest baseline under cross-site shift. On the bird domain, self-training achieves +0.12 mAP and +0.10 F1 over the strongest baseline under cross-site shift. Overall, our results show that MAST can effectively extend self-supervised audio representations from clip-level tasks to robust box-level localization across diverse bioacoustic settings, providing a practical path for biodiversity monitoring with limited labels.
\end{abstract}

\section{Introduction}

Passive acoustic monitoring (PAM) offers a scalable, non-invasive way to measure biodiversity at landscape scales \citep{pam,wood2024scalable,Lassandro03092025,burivalova2019soundscapes}. Autonomous recorders can operate for months, capturing rich soundscapes that contain vocalizations from birds, mammals, insects, amphibians, as well as environmental and anthropogenic sounds \citep{merchant2015measuring,pam_eco}. Turning these recordings into actionable ecological information typically follows two broad directions: one can use acoustic indices that summarize the properties of the soundscape, or species-specific recognition models that predict whether a known target sound occurs in a recording segment. However, the former does not identify which species produced the sounds, while the latter is typically limited to predefined classes of acoustically known species. In hyperdiverse and acoustically understudied soundscapes such as tropical rainforests \citep{sun2022classification}, an intermediate approach is needed: a system that can automatically localize animal sounds in both time and frequency, so that they can later be classified into known taxa or presented to specialists as potentially novel events. The main bottleneck is therefore not data collection, but annotation: extracting ecological signals from raw audio requires experts to mark time--frequency regions of animal sounds, which is expensive and difficult to scale \citep{NAPIER2024124220}. In addition, real-world soundscapes are highly non-stationary, with background conditions shifting across day/night cycles, weather, habitat structure, sensor placement, and species composition \citep{liang2024mind,van2024birds}. As a result, models trained on a single site or day often generalize poorly \citep{liang2024mind}.

A common strategy is to treat spectrograms as images and apply supervised detectors such as Faster R-CNN, YOLO, or FCOS \citep{eventness,hamard2024deep}. While effective in fully supervised settings, these models typically require many annotations and can overfit to local recording conditions, leading to degraded performance under temporal or spatial shift \citep{van2024birds}.

\begin{figure*}[!t]
  \centering
  \includegraphics[width=\linewidth]{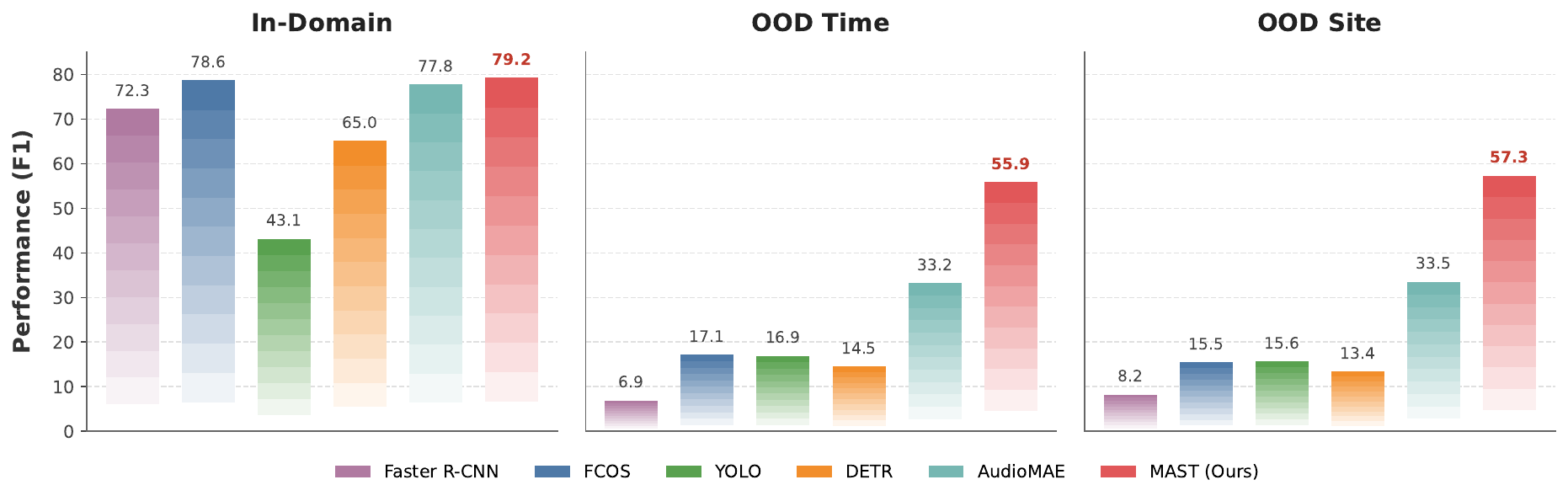}
  \vspace{-5pt}
  \caption{\textbf{Performance comparison on rainforest domain.} MAST achieves the best overall performance, with especially large improvements under temporal and cross-site distribution shift, showing stronger robustness than fully supervised baselines and generic SSL transfer.}
  \vspace{-15pt}
  \label{fig:performance}
\end{figure*}

To address these challenges, we propose MAST, a label-efficient framework for robust sound detection under limited supervision. We study a practical and stricter setting in which only a small amount of expert-labeled data, for example one fully annotated day from a single site, is available for training, while evaluation requires robust binary detection of ``any animal sound'' in terms of frequency and time across both (i) unseen days at the same site and (ii) unseen sites in the same region. Figure~\ref{fig:motivation} summarizes the motivation for MAST: rainforest soundscapes exhibit substantial temporal and cross-site distribution shift, while standard supervised detectors trained on narrow label coverage often generalize poorly. MAST addresses this by decoupling representation learning from detection. First, we adopt the ViT-based masked audio encoder architecture of AudioMAE \citep{AudioMAE} and pretrain it on large unlabeled rainforest audio via masked spectrogram reconstruction. We then transfer the pretrained encoder to a lightweight detector with audio-aware adapters and an FPN \citep{fpn} neck to recover locality and multi-scale structure. To improve separation between true events and confounding background, we add a box-level contrastive loss over region features. Finally, we exploit large unlabeled pools through iterative self-training with pseudo-labels and labeled-data refinement.

\begin{figure*}[!t]
  \centering
  \includegraphics[width=\linewidth]{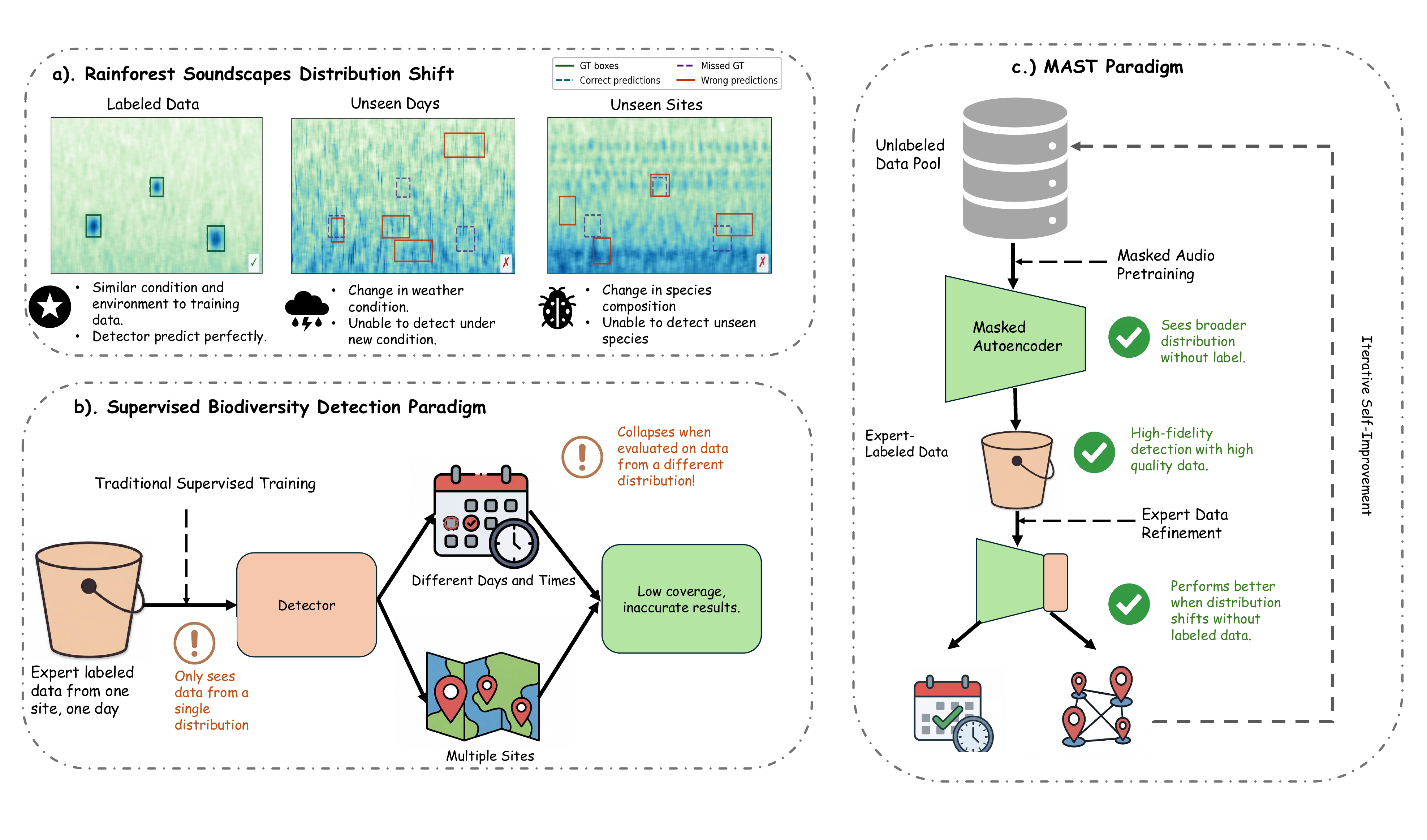}
   \vspace{-20pt}
\caption{\textbf{Motivation of MAST.} (a) Rainforest soundscapes exhibit substantial temporal and cross-site distribution shift, making detection under unseen conditions difficult. (b) Standard supervised biodiversity detectors trained on limited labels from a narrow time/site range often generalize poorly when deployed to new days and sites. (c) These challenges motivate a label-efficient and shift-robust framework that can leverage both limited expert labels and abundant unlabeled soundscapes.}
 \vspace{-15pt}
  \label{fig:motivation}
\end{figure*}

Time--frequency detection is ecologically more informative than clip-level tagging: it reveals which sounds overlap in time and frequency, how vocal activity is distributed across the spectrum, and can help determine whether detected events correspond to known or potentially novel taxa. This information is essential for biodiversity assessment but lost by coarser representations. Under comparable annotation budgets, MAST consistently improves robustness to temporal and site shift across two ecologically distinct domains: tropical rainforest soundscapes and Mediterranean bird vocalizations (Figure~\ref{fig:performance}). Our contributions are as follows:
\begin{itemize}
\item We introduce a practical pipeline for binary time--frequency sound detection that combines masked audio pretraining, adaptation, and lightweight detection under limited labels.
  \item We design a region-level contrastive loss that improves event/background separability and strengthens generalization under distribution shift.
  \item We show that an iterative self-training curriculum effectively leverages unlabeled data and improves detection in out-of-domain settings without additional expert annotation.
  \item We evaluate on two distinct bioacoustic domains, tropical rainforest soundscapes and European bird vocalizations, and provide an end-to-end recipe for converting large unlabeled recordings into robust detectors for biodiversity monitoring.
\end{itemize}

\section{Related Work}

\noindent \textbf{Passive acoustic monitoring and supervised sound detection.}
PAM enables long-duration, large-scale biodiversity recording \citep{Sueur2015-SUEETE-2,pam,CORD2025111042}, but downstream analysis remains bottlenecked by the need for manual annotation or by acoustic indices that need ground truthing \citep{kershenbaum2025automatic,bradfer2025acoustic,bradfer2023using}. Classical template-matching detectors scale poorly to diverse soundscapes \citep{Jahn2017AutomatedSoundRecognition,Towsey01062012,Katz03052016,Buxton2018AcousticIndices,barker2014automated,kershenbaum2025automatic}. Learned systems such as BirdNET \citep{birdnet} and Voxaboxen \citep{mahon2025robust} operate at clip level or along the time axis only, without full time--frequency detection. Casting the problem as object detection on spectrograms \citep{eventness,r-crnn,yoho,hamard2024deep} enables box-level localization but typically requires substantial supervision and remains brittle under site/day shift \citep{Stowell_2022,hamard2024deep,soundevent}.

\noindent \textbf{Self-supervised audio representations.}
SSL has advanced audio representation learning via masked modeling, contrastive learning, and predictive objectives \citep{wav2vec,hubert,AudioMAE,beats,niizumi2022masked,byol}. Recent bioacoustic foundation models such as Perch~2.0 \citep{van2025perch}, Bird-MAE \citep{rauch2025can}, NatureLM-audio \citep{naturelmaudio} transfer effectively to low-label monitoring tasks \citep{moummad2024self,xu2026sita}, but they operate at the clip level and do not predict time--frequency bounding boxes. Existing evaluations accordingly focus on tagging, classification, or retrieval \citep{HEAR,clap,rauch2024birdset,hagiwara2023beans}. Extending pretrained representations to dense box-level localization introduces challenges in boundary fidelity, foreground--background imbalance, and stability under non-stationary noise. Our work studies how pretrained features can be adapted for robust time--frequency detection in biodiversity soundscapes.

\noindent \textbf{Weak supervision and unlabeled soundscape learning.}
Prior SED research has explored weak supervision, semi-supervised learning, and pseudo-labeling \citep{Serizel2018,kong2020sound,wang2019comparison,park2021self,park2022cross,kim2023semi}. However, most of this literature focuses on temporal event activity rather than explicit time--frequency box detection, and stable learning from pseudo labels remains difficult under distribution shift \citep{Serizel2018,mahon2025robust}. Our method differs in both objective and setting: we use confidence-weighted pseudo-boxes within a staged self-training pipeline, where exploration on diverse pseudo-labeled soundscapes expands coverage under shift, followed by expert-guided refinement that re-anchors precision.

\begin{figure*}[!t]
  \centering
  \includegraphics[width=\linewidth]{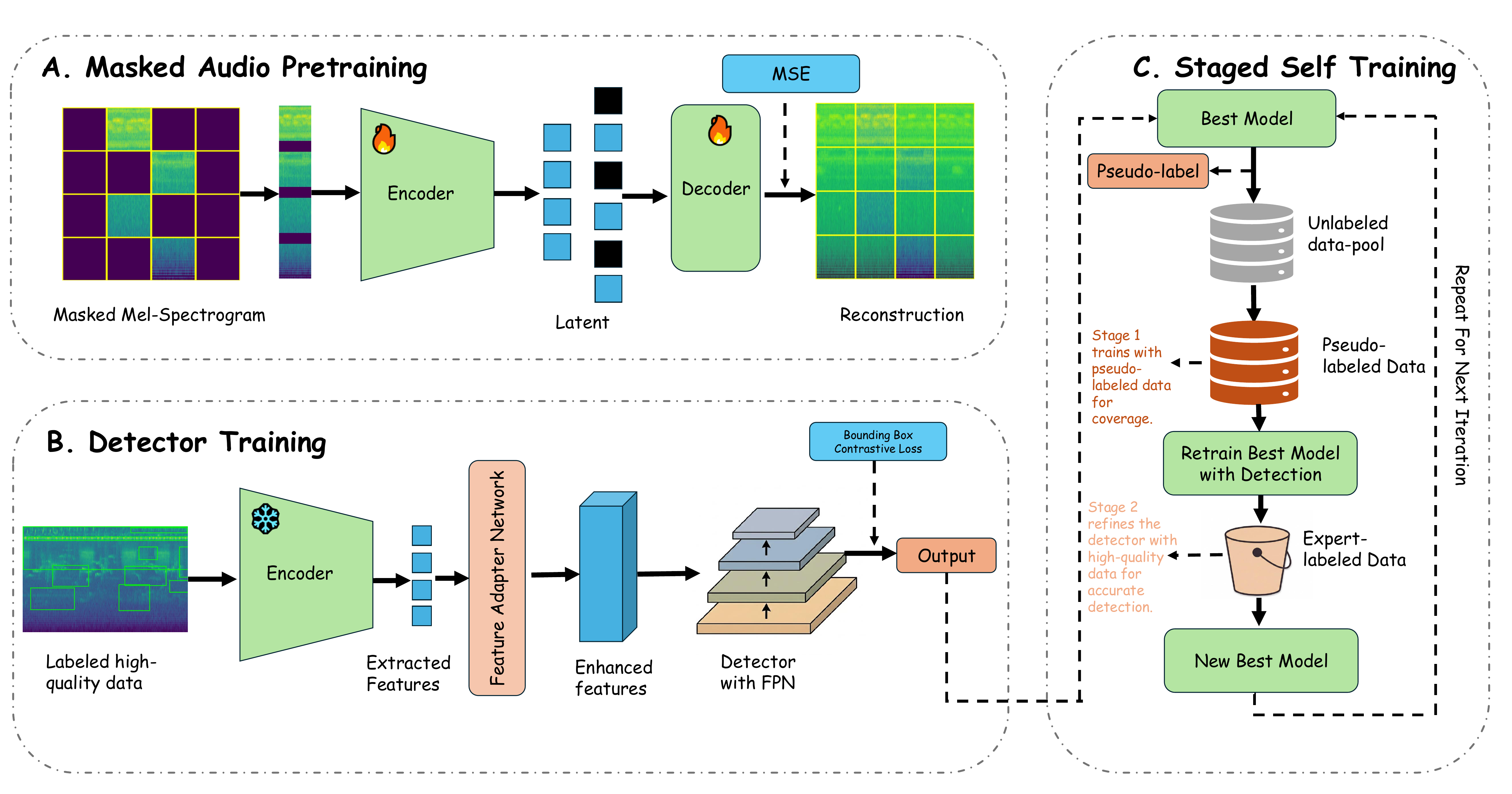}
  \vspace{-20pt}
\caption{\textbf{Overview of the MAST pipeline.} (A) Masked-audio pretraining learns general acoustic representations from unlabeled rainforest spectrograms via masked reconstruction. (B) Detector training transfers the pretrained encoder to a detection model with an audio-aware adapter and lightweight FPN, together with a box-level contrastive objective for event separation. (C) Staged self-training improves robustness under distribution shift by first generating pseudo-labels on unlabeled audio and then retraining with exploration on pseudo-labeled data followed by refinement on expert-labeled data.}
\vspace{-10pt}
  \label{fig:pipeline}
\end{figure*}

\section{Methods}
\label{sec:method}

Figure~\ref{fig:pipeline} provides an overview of MAST, which combines masked-audio pretraining, label-efficient detector adaptation, and self-training for robust sound detection under distribution shift.

\subsection{Problem Setup}

Given an audio waveform $x(t)$, we compute a log-mel spectrogram $S \in \mathbb{R}^{F \times T}$. Each clip may contain multiple annotated time--frequency boxes $B=\{(b_i, y_i)\}_{i=1}^N$, where each box $b_i=[t_i^{(1)}, t_i^{(2)}, f_i^{(1)}, f_i^{(2)}]$. Although the dataset includes sonotype labels $y_i$, our primary task is binary detection (``any animal sound''), so all annotated events are treated as a single foreground class. We consider a low-label regime with limited expert-labeled data $\mathcal{D}_{\ell}$ and abundant unlabeled data $\mathcal{D}_u$, and aim to achieve robust detection under temporal and spatial shift.

\subsection{Tasks and Objectives}

\noindent \textbf{Overview.} Our pipeline has three stages: (i) self-supervised representation learning on unlabeled soundscapes, (ii) label-efficient supervised time--frequency detection on expert boxes, and (iii) staged self-training to exploit additional unlabeled audio. See Fig.~\ref{fig:pipeline} for an overview of MAST.

\noindent \textbf{Self-supervised pretraining.} Given an unlabeled log-mel spectrogram chunk $S \sim \mathcal{D}_u$, we pretrain a masked audio encoder $E_\phi$ by masking a set of time-frequency patches $M$ and reconstructing the masked content with a lightweight decoder. We minimize the masked reconstruction loss

$$
\min _\phi \mathbb{E}_{S \sim \mathcal{D}_u}\left[\mathcal{L}_{\mathrm{MAE}}(S ; \phi)\right], \quad \mathcal{L}_{\mathrm{MAE}}=\frac{1}{|M|} \sum_{(f, t) \in M}\|\hat{S}[f, t]-S[f, t]\|_2^2 .
$$

\noindent \textbf{Supervised detection on expert labels.} For time-frequency detection, we reuse the pretrained encoder as the detector backbone and attach lightweight audio-aware adapters $A_\psi$ and a feature pyramid $N_\gamma$. A detection head $D_\theta$ predicts box coordinates and confidence scores. We train on the labeled set $\mathcal{D}_{\ell}$ with a standard detection objective:

$$
\mathcal{L}_{\mathrm{det}}=\lambda_{\mathrm{cls}} \mathcal{L}_{\mathrm{cls}}+\lambda_{\mathrm{box}} \mathcal{L}_{\mathrm{box}}+\lambda_{\mathrm{iou}} \mathcal{L}_{\mathrm{iou}},
$$

where $\mathcal{L}_{\text {cls }}$ is focal loss, $\mathcal{L}_{\text {box }}$ is $\ell_1$ regression on box corners, and $\mathcal{L}_{\text {iou }}$ is GIoU loss \citep{rezatofighi2019giou}. To improve robustness under background shift, we add a box-level contrastive term that encourages invariance across perturbations of the same event while separating events from hard background regions. For each ground-truth box we form positives via spatial/temporal jittering, and sample negatives from low-IoU proposals. Let $z_i$ be the normalized projected feature pooled from box $i$, the contrastive loss is

$$
\mathcal{L}_{\text {con }}=\frac{1}{N_{+}} \sum_i\left[-\log \frac{\sum_{p \in P_i} \exp \left(\left\langle z_i, p\right\rangle / \tau\right)}{\sum_{p \in P_i} \exp \left(\left\langle z_i, p\right\rangle / \tau\right)+\sum_{n \in N} \exp \left(\left\langle z_i, n\right\rangle / \tau\right)}\right] .
$$

The supervised training objective becomes

$$
\min _{\psi, \gamma, \theta, \phi_{\text {top }}} \mathbb{E}_{(S, B) \sim \mathcal{D}_{\ell}}\left[\mathcal{L}_{\text {det }}(S, B)+\lambda_{\text {con }} \mathcal{L}_{\text {con }}(S, B)\right],
$$

\noindent where $\phi_{\text{top}}$ denotes the parameters of the top $k$ encoder blocks that are unfrozen during fine-tuning.

\noindent \textbf{Self-training on unlabeled audio.} Starting from a seed detector trained on $\mathcal{D}_{\ell}$, we generate pseudo boxes on $\mathcal{D}_u$ using class-agnostic NMS and a confidence threshold $q$, yielding $\Pi_q\left(\mathcal{D}_u\right)$. We then retrain on the union:

$$
\min _{\psi, \gamma, \theta, \phi_{\text {top }}} \mathbb{E}_{(S, B) \sim \mathcal{D}_{\ell} \cup \Pi_q\left(\mathcal{D}_u\right)}\left[\mathcal{L}_{\text {det }}(S, B)+\lambda_{\text {con }} \mathcal{L}_{\text {con }}(S, B)\right] .
$$

To reduce confirmation bias, we apply confidence weighting that down-weights uncertain pseudo-labels and keep the backbone largely frozen in early self-training before optionally unfreezing top layers. See App.~\ref{app:self-training} for a formal definition of this stage.

\subsection{Detector Design}
\label{sec:detector}

We use the pretrained ViT-B/16 masked audio encoder $E_\phi$ as the backbone for time-frequency detection. Given a log-mel spectrogram chunk, the encoder outputs a token grid which we reshape into a feature map $E \in \mathbb{R}^{C \times H \times W}$. We then attach a lightweight audio-aware adapter $A_\psi$, a feature pyramid network $N_\gamma$, and a detection head $D_\theta$. For label-efficient transfer, we either freeze $E_\phi$ or fine-tune only the top $k$ transformer blocks, while training $A_\psi, N_\gamma$, and $D_\theta$. Our default head is an anchor-free FCOS predicting classification, centerness, and box regression. Full architectural details, adapter variants, and additional ablations are provided in App.~\ref{app:adapter},~\ref{app:architecture}, and~\ref{app:more_ablation}.

\noindent \textbf{Resolution bottleneck and asymmetric adapter.}
Transferring patch-based ViTs to dense time--frequency detection introduces a \emph{resolution bottleneck}: with patch size $p{=}16$ and spectrogram $T{\times}F = 1024{\times}128$, the encoder produces a $64{\times}8$ token grid with strides $s_t{=}s_f{=}16$. We show (Proposition~\ref{prop:iou_bound}, App.~\ref{app:adapter_theory}) that grid quantization can degrade IoU to at most $\delta_t \delta_f / [(\delta_t + s_t)(\delta_f + s_f)]$ in the worst case, and that narrow-band events with $\delta_f \leq p$ may fail to reach IoU${}\geq 0.5$ under worst-case alignment (Corollary~\ref{cor:narrow_band}). This bottleneck is asymmetric: the frequency axis has only 8 tokens, ${\sim}1$\,kHz each, while many vocalizations span $<1$\,kHz in bandwidth.

Our adapter resolves this via \emph{asymmetric learnable upsampling} ($u_t{=}2, u_f{=}4$), producing a $128{\times}32$ feature map with effective strides $8{\times}4$. The larger frequency factor addresses the more severe bottleneck. Upsampled features are refined by \emph{anisotropic} $k{\times}1$ and $1{\times}k$ convolutions that exploit the separable spectro-temporal structure of animal sounds. Empirically, this adapter improves cross-site OOD F1 from 0.348 to 0.447 and achieves the highest mean IoU (0.682) among all variants (Tables~\ref{tab:pipeline_ablation_ood},~\ref{tab:adapter_ablation_ood}). See App.~\ref{app:adapter_theory} for formal analysis with proofs.

\subsection{Self-training on Unlabeled Soundscapes}

Starting from a seed detector trained on $\mathcal{D}_{\ell}$, we generate pseudo boxes on unlabeled audio $\mathcal{D}_u$ by running the detector, applying class-agnostic NMS, and retaining predictions above a confidence threshold $q$, yielding a pseudo-labeled dataset $\Pi_q(\mathcal{D}_u)$. We then retrain with a two-stage curriculum: Stage~1 trains on $\Pi_q(\mathcal{D}_u)$ to broaden coverage across diverse backgrounds, and Stage~2 fine-tunes on $\mathcal{D}_{\ell}$ to re-anchor the model to clean expert boxes and calibrate false positives. In both stages we optimize the detection objective with our contrastive term. This cycle is repeated for 1--2 rounds. We show (Proposition~\ref{prop:dro}, App.~\ref{app:self_training_theory}) that this procedure progressively reduces the distributional gap to the target domain, and that confidence weighting provably reduces effective label noise (Proposition~\ref{prop:confidence_weighting}). See App.~\ref{app:self-training} for formal definitions and App.~\ref{app:self_training_theory} for theoretical analysis.

\section{Experiment Setup}

\noindent \textbf{Datasets and Splits.} We evaluate on two ecologically distinct domains to test cross-domain generality.

\noindent \textit{Rainforest domain.} We study sound detection on passive acoustic recordings from tropical rainforest in East Kalimantan, Indonesia, from 2017--2019. Our primary labeled dataset consists of 24 hours recorded on July 10, 2018 at one site. Audio is resampled to 16\,kHz and segmented into 10.24\,s clips; each clip is converted to a $128 \times 1024$ log-mel spectrogram (App.~\ref{app:processing}). Expert annotators provide time--frequency bounding boxes for animal sounds (vocalizations, stridulations, etc.) spanning hundreds of ``sonotypes'' - unique sound types. We collapse all sonotypes into a single foreground category and evaluate binary detection. We evaluate in two regimes: in-domain, where we do a 70/15/15 train/val/test split of the labeled site/day, and out-of-distribution (OOD), where recordings come from different days and different sites within the same landscape, introducing covariate shift in background conditions, species activity, and sensor characteristics. A large pool of unlabeled rainforest audio of $\sim$661.7\,h is used for masked audio pretraining and self-training. See App.~\ref{app:dataset_stats} for detailed statistics.

\noindent \textit{Bird domain.} To test cross-domain generality, we additionally evaluate on the BIRDeep dataset \citep{marquez2025birdeep}, which contains 641 recordings of bird vocalizations from 9 autonomous recording sites across 4 habitat types: low shrubland, high shrubland, ecotone, and marshland, in Do\~{n}ana National Park, Spain. Annotations cover 38 bird species with time--frequency bounding boxes; we again collapse to binary detection. For masked audio pretraining, we combine BIRDeep unlabeled audio with BirdSet XCM \citep{rauch2024birdset}, a large-scale bird audio collection of $\sim$900K spectrogram chunks. We construct in-domain, temporal-OOD, and site-OOD evaluation splits. See App.~\ref{app:dataset_stats} for dataset details.

\noindent \textbf{Baselines.} Recent work has applied object detectors to spectrograms for dense time--frequency detection in marine mammals \citep{hamard2024deep}, bird song \citep{eventness,yoho}, and general soundscapes \citep{soundevent}, but these rely on fully supervised detectors trained from scratch. We compare against representative architectures, all adapted to mel-spectrograms and trained on the same labeled set:
\begin{itemize}[leftmargin=13pt]
    \item \textbf{Faster R-CNN \citep{ren2016fasterrcnnrealtimeobject}}: A two-stage detector with a ResNet-50 backbone and FPN, adapted to log-mel spectrogram. We use class-agnostic detection with standard region proposal + RoIAlign.
    \item \textbf{FCOS \citep{fcos}}: An anchor-free one-stage detector with a ResNet-50 + FPN backbone, predicting per-location classification, centerness, and box regression on the spectrogram lattice.
    \item \textbf{YOLO26 \citep{yolo}}: A real-time one-stage detector. We use the YOLO26x architecture, adapted to log-mel spectrogram inputs.
    \item \textbf{DETR \citep{detr}}: A transformer-based end-to-end detector using set prediction with bipartite matching, adapted to spectrogram inputs. This baseline evaluates whether a query-based detection paradigm transfers to detection in low-label soundscapes.
    \item \textbf{AudioMAE \citep{AudioMAE}}: A transfer-learning baseline that initializes the backbone with a publicly available AudioMAE checkpoint pretrained on AudioSet, then fine-tunes the detector on our labeled split. This isolates the effect of generic large-scale audio pretraining versus in-domain self-supervised pretraining used in our method.
\end{itemize}

\begin{table}[tb!]
\caption{\textbf{In-domain detection performance} on two domains: Rainforest (tropical soundscapes) and BIRDeep (European bird vocalizations \citep{marquez2025birdeep}). \textit{Higher is better.}}
\label{tab:in_domain_results}
\centering
\small
\renewcommand{\arraystretch}{1.06}
\resizebox{0.99\linewidth}{!}{
\begin{tabular}{@{}l|ccccc|ccccc@{}}
\toprule
& \multicolumn{5}{c|}{\textbf{Rainforest}} & \multicolumn{5}{c}{\textbf{BIRDeep}} \\
Method & Prec. & Rec. & F1 & mAP & mIoU & Prec. & Rec. & F1 & mAP & mIoU \\
\midrule
Faster R-CNN \citep{ren2016fasterrcnnrealtimeobject} & 0.6832 & \textbf{0.8390} & 0.7226 & \textbf{0.7715} & 0.8590 & 0.3716 & 0.2061 & 0.2651 & 0.1537 & 0.6587 \\
FCOS \citep{fcos} & \textbf{0.9708} & 0.6601 & 0.7858 & 0.6334 & 0.9290 & 0.7571 & 0.8091 & \textbf{0.7822} & 0.6650 & 0.6849 \\
YOLO26 \citep{yolo} & 0.8659 & 0.3306 & 0.4309 & 0.3475 & 0.8102 & 0.7897 & 0.7737 & 0.7816 & 0.6183 & 0.7181 \\
DETR \citep{detr} & 0.5905 & 0.7839 & 0.6501 & 0.7422 & 0.8400 & 0.3256 & 0.7525 & 0.4545 & 0.3723 & 0.6024 \\
AudioMAE \citep{AudioMAE} & 0.9454 & 0.6903 & 0.7775 & 0.6736 & 0.9352 & \textbf{0.7951} & 0.6859 & 0.7364 & 0.6025 & 0.6756 \\
\midrule
\rowcolor{cyan!5}\textbf{MAST (ours)} & 0.9535 & 0.7021 & \textbf{0.7916} & 0.6853 & 0.9387 & 0.7824 & 0.7010 & 0.7395 & 0.6011 & 0.6822 \\
\rowcolor{cyan!5}\textbf{MAST + ST (ours)} & 0.9548 & 0.6277 & 0.7338 & 0.6277 & \textbf{0.9457} & 0.6709 & \textbf{0.9061} & 0.7709 & \textbf{0.7783} & \textbf{0.7339} \\
\bottomrule
\end{tabular}
}
\vspace{-15pt}
\end{table}

\noindent \textbf{Training Details.} All models use the same log-mel spectrogram processing and dataset-wide normalization. We train with AdamW optimizer and a warmup--cosine learning-rate schedule. For MAST detectors, we use component-wise learning rates (smaller for the backbone, larger for adapters/FPN/head) and unfreeze only the top encoder blocks during fine-tuning. We also apply lightweight spectrogram augmentations. Unless stated otherwise, we match training budgets and augmentation pipelines across methods. See App.~\ref{app:hyperparam},~\ref{app:data_augmentations} for details.

\noindent \textbf{Evaluation.} We evaluate binary time--frequency detection of ``any animal sound,'' reporting precision, recall, F1, mean average precision (mAP), and mean intersection-over-union (mIoU) at IoU $\geq 0.5$. We evaluate in two regimes: in-domain with held-out chunks from the same site/date and OOD on different days or sites. Confidence thresholds are tuned on the in-domain validation split. See App.~\ref{app:eval} for full details.

\section{Results}

\noindent \textbf{In-domain performance.}
Table~\ref{tab:in_domain_results} reports in-domain results. On rainforest, MAST achieves the best F1 of 0.792 and mean IoU of 0.939, demonstrating that self-supervised pretraining and audio-aware adaptation do not sacrifice in-domain accuracy despite being designed primarily for robustness. On BIRDeep, FCOS and YOLO26 lead in F1 at 0.782, while MAST reaches 0.740. Self-training closes this gap: MAST+ST reaches 0.771 F1 and surpasses all baselines in mAP at 0.778 vs.\ 0.665 for FCOS, in recall at 0.906, and in mean IoU at 0.734. Crucially, the remaining in-domain difference inverts under distribution shift, where MAST provides the largest gains (Table~\ref{tab:ood_combined}), highlighting that in-domain metrics alone are insufficient for field deployment.

\begin{table}[tb!]
\caption{\textbf{Out-of-distribution robustness} on Rainforest and BIRDeep. Temporal OOD evaluates on recordings from different dates at the same site. Cross-site OOD evaluates on a geographically distinct recording location. \textit{Higher is better}.}
\label{tab:ood_combined}
\centering
\small
\renewcommand{\arraystretch}{1.06}
\vspace{1pt}
\resizebox{0.99\linewidth}{!}{
\begin{tabular}{@{}l|ccccc|ccccc@{}}
\toprule
& \multicolumn{5}{c|}{\textbf{Rainforest}} & \multicolumn{5}{c}{\textbf{BIRDeep}} \\
Method & Prec. & Rec. & F1 & mAP & mIoU & Prec. & Rec. & F1 & mAP & mIoU \\
\midrule
\multicolumn{11}{c}{\textit{Temporal OOD}} \\
\midrule
Faster R-CNN \citep{ren2016fasterrcnnrealtimeobject} & 0.0400 & 0.2479 & 0.0689 & 0.0168 & 0.4259 & 0.3199 & 0.1020 & 0.1547 & 0.0909 & 0.5559 \\
FCOS \citep{fcos} & 0.5072 & 0.1029 & 0.1711 & 0.1145 & 0.6429 & 0.7269 & 0.2885 & 0.4130 & 0.2116 & 0.6510 \\
YOLO26 \citep{yolo} & \textbf{0.9459} & 0.0926 & 0.1688 & 0.0909 & \textbf{0.7603} & 0.7559 & 0.2040 & 0.3212 & 0.1520 & \textbf{0.6767} \\
DETR \citep{detr} & 0.1463 & 0.1800 & 0.1450 & 0.1316 & 0.5237 & 0.1923 & \textbf{0.6521} & 0.2971 & 0.3585 & 0.5528 \\
AudioMAE \citep{AudioMAE} & 0.8178 & 0.2085 & 0.3323 & 0.1636 & 0.6382 & \textbf{0.7959} & 0.3660 & 0.5014 & 0.2983 & 0.6328 \\
\midrule
\rowcolor{cyan!5}\textbf{MAST (ours)} & 0.7810 & 0.2968 & 0.4301 & 0.2577 & 0.6859 & 0.7618 & 0.4808 & 0.5895 & 0.3838 & 0.6083 \\
\rowcolor{cyan!5}\textbf{MAST + ST (ours)} & 0.7120 & \textbf{0.4603} & \textbf{0.5591} & \textbf{0.3340} & 0.6292 & 0.6755 & 0.6247 & \textbf{0.6491} & \textbf{0.4280} & 0.5903 \\
\midrule
\multicolumn{11}{c}{\textit{Cross-site OOD}} \\
\midrule
Faster R-CNN \citep{ren2016fasterrcnnrealtimeobject} & 0.0497 & 0.2341 & 0.0820 & 0.0281 & 0.4223 & 0.8824 & 0.0744 & 0.1372 & 0.0909 & 0.5863 \\
FCOS \citep{fcos} & 0.5848 & 0.0892 & 0.1548 & 0.0909 & 0.6513 & \textbf{0.9725} & 0.2852 & 0.4410 & 0.2592 & 0.6566 \\
YOLO26 \citep{yolo} & \textbf{0.9910} & 0.0849 & 0.1564 & 0.0909 & 0.7484 & 0.8889 & 0.1934 & 0.3177 & 0.1675 & \textbf{0.7010} \\
DETR \citep{detr} & 0.1269 & 0.1758 & 0.1338 & 0.1130 & 0.5607 & 0.1334 & \textbf{0.6621} & 0.2220 & 0.3032 & 0.5554 \\
AudioMAE \citep{AudioMAE} & 0.8028 & 0.2113 & 0.3346 & 0.1742 & 0.6677 & 0.7266 & 0.4284 & 0.5390 & 0.3217 & 0.6072 \\
\midrule
\rowcolor{cyan!5}\textbf{MAST (ours)} & 0.7562 & 0.3169 & 0.4466 & 0.2573 & \textbf{0.6824} & 0.6476 & 0.5059 & 0.5680 & 0.3803 & 0.6001 \\
\rowcolor{cyan!5}\textbf{MAST + ST (ours)} & 0.6978 & \textbf{0.4838} & \textbf{0.5715} & \textbf{0.3921} & 0.6420 & 0.6522 & 0.6438 & \textbf{0.6434} & \textbf{0.4414} & 0.5943 \\
\bottomrule
\end{tabular}
}
\vspace{-15pt}
\end{table}

\noindent \textbf{Out-of-distribution robustness.}
Table~\ref{tab:ood_combined} reports OOD results under temporal and cross-site shift. Across both domains, detectors trained from scratch degrade sharply. For example, FCOS drops from 0.786 to 0.155 cross-site F1 on rainforest, a $5{\times}$ collapse. In contrast, MAST achieves the strongest overall F1 and mAP. On the rainforest domain, MAST raises cross-site F1 from 0.335 to 0.447 and temporal F1 from 0.332 to 0.430 over the strongest transfer baseline AudioMAE. On BIRDeep, MAST improves temporal F1 from 0.501 to 0.590 and cross-site F1 from 0.539 to 0.568 over AudioMAE, confirming that domain-matched pretraining is more transferable than generic AudioSet initialization. Notably, most baselines exhibit a sharp precision--recall imbalance under shift: FCOS and YOLO26 retain high precision but collapse in recall, while DETR and Faster~R-CNN maintain recall at the cost of precision. MAST is the only method that sustains both, yielding the highest F1 across conditions. The consistent pattern across both ecologically distinct domains, with modest in-domain gaps but substantial OOD improvements, confirms that MAST's primary benefit is robustness to distribution shift, the key deployment challenge in real-world biodiversity monitoring. Self-training further boosts rainforest OOD F1 to 0.559 temporal and 0.572 cross-site, and BIRDeep OOD F1 to 0.649 temporal and 0.643 cross-site.

\begin{figure}[t]
\centering
\includegraphics[width=\linewidth]{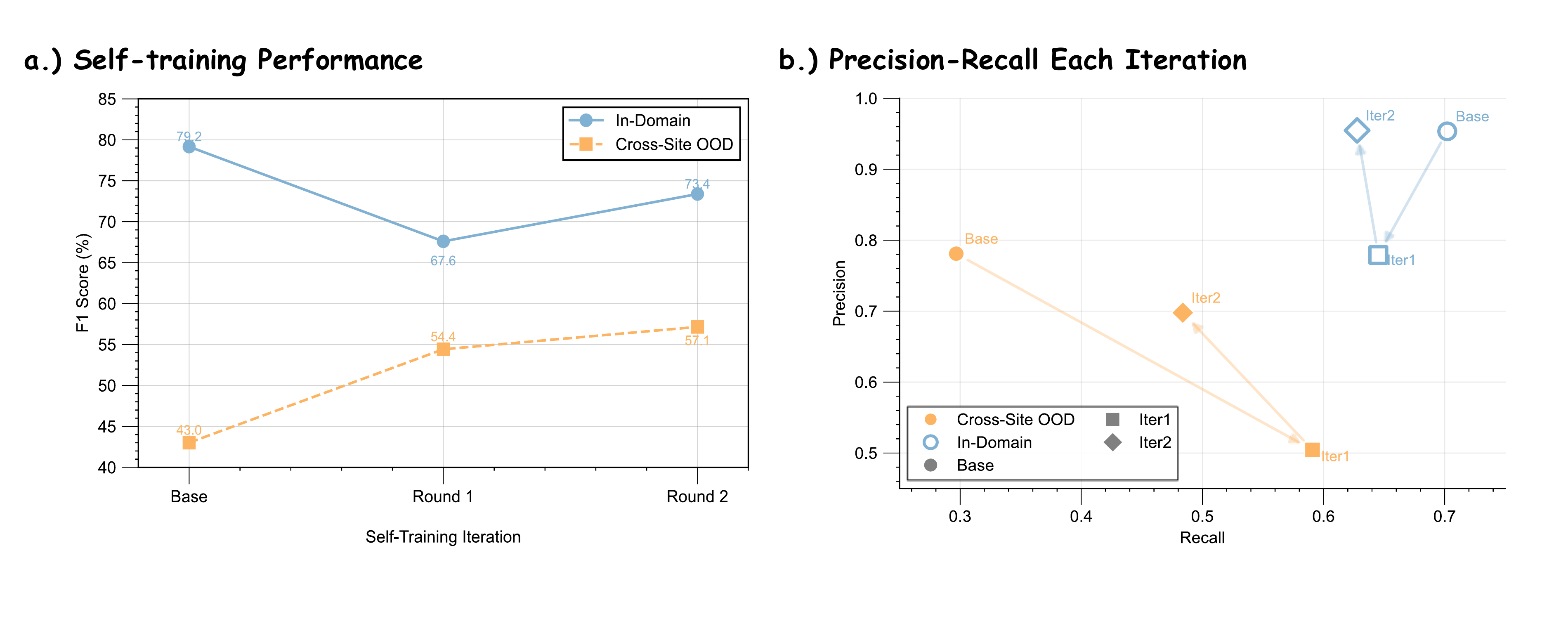}
\vspace{-20pt}
\caption{Effect of iterative self-training on F1 for rainforest domain. Round~0 is the best MAST detector after expert label refinement. Rounds~1--2 are models trained with the two-stage curriculum. OOD F1 improves monotonically from 0.4466 to 0.5442 and 0.5715, while in-domain F1 drops after Round~1 from 0.7916 to 0.6760 and partially recovers after Round~2 to 0.7338.}
\vspace{-16pt}
\label{fig:self-training}
\end{figure}

\noindent \textbf{Effect of self-training.}
Figure~\ref{fig:self-training} shows that OOD F1 improves monotonically from 0.447 at Round~0 to 0.544 after one self-training round and to 0.572 after two rounds, a +0.13 absolute gain, while in-domain F1 dips after Round~1 and partially recovers after Round~2. The in-domain dip is expected: pseudo-label training broadens the decision boundary to cover unseen acoustic conditions, temporarily trading in-domain precision for OOD coverage, while the Stage~2 expert refinement partially corrects this drift. The first round provides the largest jump, suggesting that pseudo-labeling mainly helps by exposing the detector to diverse background conditions absent from the labeled site/day. The monotonic OOD trend, achieved without any additional expert annotation, supports self-training as a practical mechanism for improving cross-condition coverage. On BIRDeep, self-training similarly improves OOD F1 from 0.568 to 0.643 cross-site and from 0.590 to 0.649 temporal, with especially large mAP gains from 0.380 to 0.441 cross-site, confirming that the two-stage curriculum generalizes across domains. Per-round results are in App.~\ref{app:birdeep_st_rounds}. The consistent gains on both domains suggest that further rounds or larger unlabeled pools could yield additional improvements.

\begin{table}[tb!]
\caption{\textbf{Pipeline ablation on cross-site OOD detection on rainforest domain.}
Starting from the full MAST model, we progressively remove components to measure their individual contributions. We also report an AudioSet-pretrained variant and the FCOS baseline for reference. \textit{Higher is better}.}
\label{tab:pipeline_ablation_ood}
\centering
\small
\renewcommand{\arraystretch}{1.08}
\resizebox{0.99\linewidth}{!}{
\begin{tabular}{@{}l|ccccc@{}}
\toprule
\textbf{Variant} & \textbf{Precision} & \textbf{Recall} & \textbf{F1} & \textbf{mAP} & \textbf{Mean IoU} \\
\midrule
\rowcolor{cyan!5}\textbf{MAST (full)} & 0.7562 & 0.3169 & \textbf{0.4466} & \textbf{0.2573} & \textbf{0.6824} \\
$-$\; Contrastive loss & 0.6536 & \textbf{0.3327} & 0.4410 & 0.2434 & 0.6555 \\
$-$\; Audio-aware adapter & 0.6380 & 0.2397 & 0.3484 & 0.2237 & 0.6607 \\
\;\; Replace with AudioSet pretraining & \textbf{0.8028} & 0.2113 & 0.3346 & 0.1742 & 0.6677 \\
$-$\; All pretraining (from scratch) & 0.3607 & 0.3666 & 0.3636 & 0.1840 & 0.5588 \\
\midrule
FCOS baseline & 0.5848 & 0.0892 & 0.1548 & 0.0909 & 0.6513 \\
\bottomrule
\end{tabular}
}
\vspace{-15pt}
\end{table}

\noindent \textbf{Ablation studies.}
Table~\ref{tab:pipeline_ablation_ood} decomposes MAST's gains under cross-site OOD shift by progressively removing components from the full model before self-training. Removing the contrastive loss drops precision sharply (0.756$\to$0.654) while recall rises slightly, leaving F1 nearly unchanged; the loss thus acts primarily as a confidence calibrator that sharpens the decision boundary, producing higher-quality pseudo-labels for self-training. Removing the adapter further reduces F1 from 0.441 to 0.348, demonstrating the importance of resolving the resolution bottleneck (Sec.~\ref{sec:detector}). Replacing domain-matched pretraining with generic AudioSet pretraining yields lower mAP, 0.174 vs.\ 0.224, confirming the value of in-domain representations. Removing all pretraining entirely still outperforms the CNN baseline in F1, 0.364 vs.\ 0.155, suggesting that the ViT's self-attention can model long-range spectro-temporal dependencies that aid generalization under background shift, even without pretrained weights. Each component contributes a meaningful and complementary gain; see App.~\ref{app:more_ablation} for additional adapter comparisons.

Table~\ref{tab:self_training_ablation} ablates the self-training curriculum. Stage~1 or Stage~2 alone both collapse recall despite high precision, yielding poor F1, demonstrating that naive single-stage pseudo-label training is insufficient. A one-stage union achieves high recall and mAP but precision drops to 0.157, indicating that mixing noisy pseudo supervision with expert labels without calibration leads to degenerate behavior. Only the full two-stage curriculum achieves the best F1 of 0.544 by balancing coverage and precision: Stage~1 broadens coverage under shift, while Stage~2 re-anchors the detector to clean expert annotations. This validates the curriculum design as essential rather than incidental to self-training success.

\begin{table}[tb!]
\caption{\textbf{Ablation of the self-training curriculum on cross-site OOD detection on rainforest domain.}
We compare the MAST model against Stage~1-only, Stage~2-only, a one-stage union of expert-labeled and pseudo-labeled data, and the full two-stage curriculum. \textit{Higher is better}.}
\label{tab:self_training_ablation}
\centering
\footnotesize
\renewcommand{\arraystretch}{1.05}
\begin{tabular}{@{}lccccc@{}}
\toprule
\textbf{Variant} & \textbf{Precision} & \textbf{Recall} & \textbf{F1} & \textbf{mAP} & \textbf{Mean IoU} \\
\midrule
Base (no ST) & 0.7562 & 0.3169 & 0.4466 & 0.2573 & 0.6824 \\
Stage 1 only & 0.9714 & 0.1201 & 0.2138 & 0.1731 & 0.7439 \\
Stage 2 only & \textbf{0.9971} & 0.1046 & 0.1893 & 0.1789 & \textbf{0.7644} \\
One stage union  & 0.1572 & \textbf{0.7341} & 0.2590 & \textbf{0.3894} & 0.5525 \\
\midrule
\rowcolor{cyan!5}\textbf{Stage 1 + Stage 2 (Full)} & 0.5044 & 0.5909 & \textbf{0.5442} & 0.3864 & 0.6067 \\
\bottomrule
\end{tabular}
\vspace{-15pt}
\end{table}

\noindent \textbf{Retrieval-based event characterization.}
A practical benefit of MAST's design is that the pretrained encoder naturally provides meaningful box-level embeddings: each detected event can be represented by pooling the corresponding time--frequency region from the feature map. This enables a retrieval-based workflow where detected sounds are matched to annotated examples via nearest-neighbor lookup for zero-shot classification, or grouped by acoustic similarity to discover candidate novel sounds for expert review. On oracle ground-truth boxes, retrieval achieves 72.3\% Top-1 accuracy over 131 sonotypes; on high-confidence predicted boxes with IoU~$\geq 0.5$, Top-1 reaches 96.1\%, indicating that detections align with acoustically prototypical events. This allows ecologists to perform species-level analysis directly on detector outputs without retraining a multi-class model. Full results are in App.~\ref{app:box_embedding_application}.

\noindent \textbf{Label budget and sampling strategy.}
A practical question for bioacousticians deploying a new detector is: \emph{how much annotation effort is needed, and how should it be allocated?} We study this by training the FCOS baseline on varying numbers of labeled chunks under three sampling strategies. Both random and time-balanced sampling substantially outperform sequential labeling, especially at low budgets where the gap reaches +0.13 F1 at $N{=}1000$ chunks, because sequential annotation under-represents acoustic conditions outside the labeled time window. Time-balanced sampling performs slightly better than random at most budgets, and as the budget grows all strategies converge. Since sequential labeling is the most natural default, where annotators typically start at the beginning of a recording and label forward in time, these results serve as a practical warning: practitioners should prioritize temporal diversity over volume when annotation resources are limited. Full results and protocol details are in App.~\ref{app:label_budget}.

\section{Conclusion}

We introduced MAST, a label-efficient framework for biodiversity sound detection that combines masked-audio pretraining, audio-aware detection adaptation, and iterative self-training. Using expert labels from only a single annotated day at one site, MAST substantially outperforms supervised baselines and generic AudioSet transfer under both temporal and cross-site distribution shift across two ecologically distinct domains, tropical rainforest soundscapes and Mediterranean bird vocalizations. Domain-matched self-supervised pretraining proves more transferable than large-scale generic pretraining, the audio-aware adapter resolves a critical resolution bottleneck for narrow-band vocalizations, and the two-stage self-training curriculum yields monotonic OOD gains without additional expert annotation. Together, these results demonstrate a scalable path from large unlabeled acoustic archives to robust time--frequency detectors under realistic label-limited conditions, supporting biodiversity assessment tasks currently bottlenecked by manual analysis, such as tracking changes in vocal activity across habitats, flagging acoustically novel events for taxonomic review, and scaling soundscape monitoring to landscape-level surveys. Limitations and future directions are discussed in App.~\ref{app:limitations}.

\newpage

\bibliographystyle{plainnat}
\bibliography{neurips_2026}

\newpage
\appendix

\section{Limitations and Future Work}
\label{app:limitations}

\noindent \textbf{Limitations.}
MAST performs binary ``any animal sound'' detection and does not distinguish among species or call types, though the retrieval workflow (App.~\ref{app:box_embedding_application}) partially bridges this gap. Self-training is bounded by the seed detector's quality: systematic false negatives cannot be recovered through pseudo-labeling alone. The pseudo-label confidence threshold is tuned on the in-domain validation split, which may not be optimal under OOD conditions; adaptive or domain-aware thresholding could further improve self-training. Finally, the full pipeline requires substantially more compute than a single supervised detector.

\noindent \textbf{Future work.}
Extending MAST to multi-class time--frequency detection via species-level annotations or unsupervised clustering of detected events is a natural next step. Incorporating complementary weak signals such as clip-level tags or acoustic indices could help break the pseudo-label quality ceiling. We evaluate on two ecologically distinct domains, but validation on additional biomes such as marine and urban and taxa like marine mammals and amphibians would further establish generality.

\section{Self-training Operator and Two-Stage Curriculum}
\label{app:self-training}

Given an unlabeled chunk $S$, a detector produces a set of scored candidate boxes

$$
\widehat{\mathcal{B}}(S)=\left\{\left(\hat{b}_j, s_j\right)\right\}_{j=1}^{\widehat{N}}, \quad s_j \in[0,1] .
$$

We apply class-agnostic NMS with overlap threshold $\eta$, and enforce a validity predicate $v(\hat{b})$, such as minimum time span, minimum frequency span, and boundary constraints. With a confidence policy $q(\cdot)$, the pseudo-label operator is

$$
\Pi_q(S)=\left\{\hat{b} \in \operatorname{NMS}_\eta(\widehat{\mathcal{B}}(S)): s(\hat{b}) \geq q(S) \wedge v(\hat{b})\right\} .
$$

Applying $\Pi_q$ to the unlabeled pool yields a pseudo-labeled dataset

$$
\Pi_q\left(\mathcal{D}_u\right)=\left\{(S, \tilde{B}): S \in \mathcal{D}_u, \tilde{B}=\Pi_q(S)\right\}.
$$

 In our experiments we set $q$ on the labeled validation split to balance precision and coverage. Remaining noise is suppressed by confidence weighting during training (see Stage~1 below).

We then train student detectors using two stages:

\noindent \textbf{Stage 1: pseudo pretraining.} We train on $\Pi_q\left(\mathcal{D}_u\right)$ using the detection loss with the box-level contrastive loss. To reduce the impact of noisy pseudo labels, we weight each pseudo-labeled example by its average pseudo confidence:

$$
\bar{g}(S)=\frac{1}{|\tilde{B}|} \sum_{\hat{b} \in \tilde{B}} s(\hat{b}),
$$

and optimize

$$
\mathcal{L}_{\text {stage1 }}=\mathbb{E}_{(S, \tilde{B}) \sim \Pi_q\left(\mathcal{D}_u\right)}\left[\bar{g}(S) \mathcal{L}_{\text {det }}(S, \tilde{B})+\lambda_{\text {con }} \mathcal{L}_{\text {con }}(S, \tilde{B})\right]
$$

Intuitively, Stage 1 increases recall and robustness by exposing the detector to diverse acoustic conditions that are absent from $\mathcal{D}_{\ell}$, while confidence weighting reduces confirmation bias from uncertain predictions.

\noindent \textbf{Stage 2: expert refinement.}
We initialize from the Stage 1 student and fine-tune on the expert-labeled set:

$$
\mathcal{L}_{\text {stage2 }}=\mathbb{E}_{(S, B) \sim \mathcal{D}_{\ell}}\left[\mathcal{L}_{\text {det }}(S, B)+\lambda_{\text {con }} \mathcal{L}_{\text {con }}(S, B)\right] .
$$

Stage 2 re-anchors training to high-quality annotations, mitigating confirmation bias from pseudo labels and improving calibration under domain shift.

We then repeat the cycle of pseudo-labeling $\to$ Stage 1 $\to$ Stage 2 for 1--2 rounds depending on unlabeled pool size and observed saturation.

\section{Theoretical Analysis of Adapter Design}
\label{app:adapter_theory}

We formalize well-known intuitions about the resolution bottleneck induced by patch-based vision transformers \citep{li2022vitdet,fpn}, derive the optimality of asymmetric upsampling for our setting, and analyze the efficiency of anisotropic convolutions.

\subsection{Resolution Bottleneck}

With patch size $p$ and spectrogram dimensions $T \times F$, the encoder produces a feature map with strides $s_t = p$ and $s_f = p$, yielding a token grid of size $(T/p) \times (F/p)$. A detector operating on this grid must express box boundaries at the granularity of these strides. Prior work has noted that coarse feature-map strides hurt small-object localization \citep{fpn,li2022vitdet}, motivating multi-scale feature pyramids. Here we make this intuition precise for the spectrogram setting and show that quantization can severely degrade localization quality, particularly along the frequency axis.

\begin{proposition}[Quantization-induced IoU degradation]
\label{prop:iou_bound}
Let $b^* = [t_1, f_1, t_2, f_2]$ be a ground-truth box with temporal extent $\delta_t = t_2 - t_1$ and frequency extent $\delta_f = f_2 - f_1$. Let $\hat{b}$ denote the tightest enclosing box whose corners lie on a grid with strides $s_t$ and $s_f$. Then there exists a placement of $b^*$ relative to the grid such that:
$$
\mathrm{IoU}(b^*, \hat{b}) \leq \frac{\delta_t \cdot \delta_f}{(\delta_t + s_t)(\delta_f + s_f)}.
$$
\end{proposition}

\begin{proof}
We construct a worst-case placement as follows. Choose $t_1 = (m{+}1)s_t - \epsilon$ for integer $m$ and small $\epsilon > 0$, so $\lfloor t_1/s_t \rfloor = m$. Write $\delta_t = k_t s_t + r_t$ with $0 \leq r_t < s_t$. If $r_t = 0$, then $t_2 = (m{+}k_t{+}1)s_t - \epsilon$, so $\lceil t_2/s_t \rceil = m + k_t + 1$ and $\hat{\delta}_t = (k_t{+}1)s_t = \delta_t + s_t$. If $r_t > 0$ (choosing $\epsilon < r_t$), then $t_2 = (m{+}k_t{+}1)s_t + r_t - \epsilon$ with $r_t - \epsilon \in (0, s_t)$, so $\lceil t_2/s_t \rceil = m + k_t + 2$ and $\hat{\delta}_t = (k_t{+}2)s_t = \delta_t + 2s_t - r_t > \delta_t + s_t$. In both cases $\hat{\delta}_t \geq \delta_t + s_t$. The same construction applies in frequency. Since $b^* \subseteq \hat{b}$, $\mathrm{IoU}(b^*, \hat{b}) = \delta_t \delta_f / (\hat{\delta}_t \hat{\delta}_f) \leq \delta_t \delta_f / [(\delta_t + s_t)(\delta_f + s_f)]$.
\end{proof}

\begin{corollary}[Narrow-band events can fall below detection threshold]
\label{cor:narrow_band}
Standard evaluation counts a prediction as a true positive only if $\mathrm{IoU} \geq 0.5$. With encoder grid strides $s_t = s_f = p = 16$ and a narrow-band vocalization with $\delta_f \leq p$ (frequency extent $\leq 1$ token), the IoU bound from Proposition~\ref{prop:iou_bound} is monotonically increasing in $\delta_f$, so it is maximized at $\delta_f = p$, giving:
$$
\mathrm{IoU}(b^*, \hat{b}) \leq \frac{\delta_t \cdot p}{(\delta_t + p)(p + p)} = \frac{\delta_t}{2(\delta_t + p)} < \frac{1}{2},
$$
for all finite $\delta_t$. By Proposition~\ref{prop:iou_bound}, there exists a placement achieving this bound, hence \textbf{there exist placements where the tightest enclosing grid-aligned box fails to reach IoU $\geq 0.5$} for any narrow-band event, regardless of its temporal extent.
\end{corollary}

In our setting, each token spans $\Delta t_{\mathrm{tok}} = p \cdot h = 160\,\text{ms}$ in time, where $h = 10\,\text{ms}$ is the STFT hop, and $\Delta f_{\mathrm{tok}} = p \cdot \Delta f_{\mathrm{mel}} \approx 1\,\text{kHz}$ in frequency. Many animal vocalizations such as bird harmonics, typically with 200--500\,Hz bandwidth, and insect stridulations, typically with $<1$\,kHz, occupy at most one token in frequency while spanning many tokens in time, falling precisely into the regime of Corollary~\ref{cor:narrow_band}.

\subsection{Asymmetric Upsampling: Theoretical Analysis}

Our adapter applies learnable upsampling with asymmetric factors $u_t$ in time and $u_f$ in frequency, reducing the effective strides to $s_t' = p/u_t$ and $s_f' = p/u_f$. By the same worst-case analysis as Proposition~\ref{prop:iou_bound}:
$$
\mathrm{IoU}(b^*, \hat{b}) \leq \frac{\delta_t \cdot \delta_f}{\left(\delta_t + p/u_t\right)\left(\delta_f + p/u_f\right)}.
$$

\begin{proposition}[Optimal upsampling allocation]
\label{prop:optimal_upsample}
Given a fixed upsampling budget $u_t \cdot u_f = U$ (total spatial expansion), the IoU bound in Proposition~\ref{prop:iou_bound} is maximized when the upsampling factors are allocated proportionally to the stride-to-extent ratios:
$$
\frac{u_f}{u_t} = \frac{s_f / \delta_f}{s_t / \delta_t} = \frac{\delta_t}{\delta_f}.
$$
That is, the axis with greater relative quantization error should receive more upsampling.
\end{proposition}

\begin{proof}
Since the numerator $\delta_t \delta_f$ is constant, maximizing $\mathrm{IoU}_{\mathrm{ub}}$ is equivalent to minimizing the denominator $D(u_t) = (\delta_t + p/u_t)(\delta_f + p u_t/U)$ after substituting $u_f = U/u_t$. Setting $D'(u_t) = 0$:
$$
\frac{-p}{u_t^2}\!\left(\delta_f + \frac{p u_t}{U}\right) + \left(\delta_t + \frac{p}{u_t}\right)\!\frac{p}{U} = 0
\quad\Longrightarrow\quad
\frac{\delta_f + p u_t/U}{u_t^2} = \frac{\delta_t + p/u_t}{U}.
$$
Cross-multiplying and expanding: $U\delta_f + p u_t = u_t^2 \delta_t + p u_t$. The $p u_t$ terms cancel, giving $U\delta_f = u_t^2 \delta_t$, so $u_t = \sqrt{U \delta_f / \delta_t}$ and $u_f = U/u_t = \sqrt{U \delta_t / \delta_f}$, hence $u_f/u_t = \delta_t/\delta_f$. This is a maximum since $D(u_t) \to \infty$ as $u_t \to 0^+$ or $u_t \to \infty$.
\end{proof}

For typical rainforest vocalizations with $\delta_t \gg \delta_f$ (temporally extended, frequency-narrow), this predicts $u_f > u_t$, which is exactly our design choice ($u_t = 2, u_f = 4$, ratio $u_f/u_t = 2$). With these factors, the adapted feature map has size $128 \times 32$ and effective strides of $8 \times 4$ pixels. For a narrow-band event with $\delta_f = p = 16$, the IoU upper bound improves from $< 0.5$ (Corollary~\ref{cor:narrow_band}) to $\delta_t \cdot 16 / [(\delta_t + 8) \cdot 20]$, which exceeds 0.5 for $\delta_t \geq 14$ (events longer than ${\sim}$140\,ms).

\subsection{Anisotropic Convolution: Separability Analysis}

\begin{proposition}[Parameter efficiency of separable filtering]
\label{prop:separable}
Let $\mathcal{F}_{\mathrm{iso}}: \mathbb{R}^{C \times H \times W} \to \mathbb{R}^{C \times H \times W}$ be an isotropic convolution with kernel size $k \times k$, and $\mathcal{F}_{\mathrm{sep}} = \mathcal{F}_{1 \times k} \circ \mathcal{F}_{k \times 1}$ be the composition of two anisotropic convolutions. Then:
\begin{enumerate}[label=(\roman*)]
    \item \textbf{Parameter reduction:} $|\Theta_{\mathrm{sep}}| = 2C^2 k$ versus $|\Theta_{\mathrm{iso}}| = C^2 k^2$, a factor of $k/2$ savings.
    \item \textbf{Receptive field preservation:} Both achieve effective receptive field $k \times k$ on the input.
    \item \textbf{Rank constraint:} The effective spatial kernel of $\mathcal{F}_{\mathrm{sep}}$ for each input--output channel pair has rank at most $\min(C, k)$, whereas a general $k \times k$ filter is unconstrained and can achieve full rank $k$. The composition processes time and frequency axes sequentially, imposing a structural separability that acts as beneficial inductive bias for bioacoustic events with approximately independent temporal and spectral structure.
\end{enumerate}
\end{proposition}

\begin{proof}
(i) The $k \times 1$ filter has $C^2 k$ parameters and the $1 \times k$ filter has $C^2 k$ parameters, totaling $2C^2 k$. The ratio is $C^2 k^2 / (2C^2 k) = k/2$.
(ii) The $k \times 1$ filter sees $k$ rows and 1 column; the subsequent $1 \times k$ filter sees 1 row and $k$ columns of the intermediate map, so each output pixel depends on a $k \times k$ neighborhood.
(iii) For a fixed channel pair $(c_{\mathrm{in}}, c_{\mathrm{out}})$, the effective $k \times k$ spatial kernel is $K[i,j] = \sum_{c_m=1}^{C} W_f[c_{\mathrm{out}}, c_m, j] \cdot W_t[c_m, c_{\mathrm{in}}, i]$, a sum of $C$ rank-1 outer products. Hence $\mathrm{rank}(K) \leq \min(C, k)$, whereas a general $k \times k$ filter is unconstrained and can achieve full rank $k$. More fundamentally, $\mathcal{F}_{\mathrm{sep}}$ processes time and frequency axes sequentially rather than jointly, so it cannot model non-separable cross-axis interactions. For signals where temporal and spectral structure are approximately independent, as in harmonic animal vocalizations with smooth temporal envelopes, this separability constraint matches the signal structure and suppresses spurious cross-axis correlations.
\end{proof}

This analysis explains why the anisotropic adapter achieves the best F1 among adapter variants (Table~\ref{tab:adapter_ablation_ood}): the separability constraint matches the physical structure of bioacoustic events, providing effective regularization without sacrificing expressiveness for the target signal class.

\section{Theoretical Analysis of Self-Training}
\label{app:self_training_theory}

Motivated by domain adaptation theory \citep{ben2010theory} and recent analyses of self-training under distribution shift \citep{kumar2020understanding,wei2020theoretical}, we provide theoretical justification for MAST's two-stage self-training curriculum, confidence-weighted pseudo-label training, and the improvement conditions under iterative self-training. The following results adapt established frameworks to our specific setting of two-stage pseudo-label exploration and expert refinement for time--frequency detection.

\subsection{Self-Training as Distributionally Robust Optimization}

Let $P_\ell$ denote the labeled source distribution and $P_u$ denote the unlabeled target distribution, which may include temporal and cross-site domain shift. The standard supervised objective minimizes $\mathbb{E}_{(S,B) \sim P_\ell}[\mathcal{L}(f_\theta(S), B)]$, but the deployed detector must perform well under $P_u$. Following the distributional robustness perspective of \citet{ben2010theory}, we characterize our two-stage self-training as approximate DRO that progressively reduces the gap between training and deployment distributions.

\begin{proposition}[Self-training reduces distributional gap]
\label{prop:dro}
Let $f^{(0)}$ be a detector trained on $P_\ell$, and let $\hat{P}_u^{(t)} = \Pi_q(P_u; f^{(t)})$ denote the pseudo-labeled distribution induced by applying detector $f^{(t)}$ to unlabeled data with confidence threshold $q$. Define the \emph{effective training distribution} at round $t$ of the two-stage curriculum as:
$$
P_{\mathrm{eff}}^{(t)} = (1 - \alpha) \hat{P}_u^{(t)} + \alpha \, P_\ell,
$$
where $\alpha \in (0,1)$ reflects the relative contribution of Stage~2. Let $d_{\mathrm{TV}}(\cdot, \cdot)$ denote total variation distance. If the pseudo-labeling operator has precision $\pi_q = \Pr[\hat{b} \text{ correct} \mid s(\hat{b}) \geq q]$ and the detector achieves coverage $\rho_q = \Pr_{b^* \sim P_u}[\exists \hat{b}: \mathrm{IoU}(\hat{b}, b^*) \geq 0.5 \wedge s(\hat{b}) \geq q]$, then:
$$
d_{\mathrm{TV}}(P_{\mathrm{eff}}^{(t)}, P_u) \leq (1 - \alpha)(1 - \rho_q \pi_q) + \alpha \, d_{\mathrm{TV}}(P_\ell, P_u).
$$
\end{proposition}

\begin{proof}
Since $\hat{P}_u^{(t)}$ and $P_u$ share the same marginal on $S$ (pseudo-labeling only modifies the annotation), their TV distance is determined by label disagreement. The pseudo-labeled distribution agrees with $P_u$ on the fraction $\rho_q \pi_q$ of events that are both covered and correctly labeled. The remaining mass $(1-\rho_q\pi_q)$ may have incorrect or missing labels, contributing at most $(1-\rho_q\pi_q)$ to $d_{\mathrm{TV}}(\hat{P}_u^{(t)}, P_u)$. By convexity of total variation:
\begin{align*}
d_{\mathrm{TV}}(P_{\mathrm{eff}}^{(t)}, P_u) &\leq (1-\alpha) \, d_{\mathrm{TV}}(\hat{P}_u^{(t)}, P_u) + \alpha \, d_{\mathrm{TV}}(P_\ell, P_u) \\
&\leq (1-\alpha)(1 - \rho_q \pi_q) + \alpha \, d_{\mathrm{TV}}(P_\ell, P_u).
\end{align*}
\end{proof}

This result shows that as pseudo-label quality ($\pi_q$) and coverage ($\rho_q$) improve across self-training rounds, $P_{\mathrm{eff}}^{(t)}$ converges toward $P_u$, reducing the effective domain gap. The two-stage design is crucial: Stage~1 improves $\rho_q$, the coverage via diverse pseudo-labeled data, while Stage~2 controls $\alpha$, the re-anchoring to clean labels prevents drift.

\subsection{Confidence Weighting as Importance-Weighted Risk Minimization}

In Stage~1, we weight each pseudo-labeled sample by its average pseudo confidence $\bar{g}(S) = \frac{1}{|\tilde{B}|} \sum_{\hat{b} \in \tilde{B}} s(\hat{b})$. Building on the importance-weighted learning framework for noisy labels \citep{natarajan2013learning}, we observe that this weighting can be interpreted as importance-weighted empirical risk minimization that provably reduces the effective noise rate.

\begin{proposition}[Confidence weighting reduces effective label noise]
\label{prop:confidence_weighting}
Let $\epsilon(s)$ denote the label error probability for a pseudo-label with confidence $s$, and assume $\epsilon(s)$ is monotonically decreasing in $s$ (higher confidence $\Rightarrow$ lower error rate). Given $n$ pseudo-labeled samples $(S_i, \tilde{B}_i)_{i=1}^n$, the confidence-weighted empirical risk
$$
\hat{R}_w(\theta) = \frac{1}{Z} \sum_{i=1}^n \bar{g}(S_i) \, \mathcal{L}(f_\theta(S_i), \tilde{B}_i), \quad Z = \sum_{i=1}^n \bar{g}(S_i),
$$
has effective noise rate
$$
\bar{\epsilon}_w = \frac{\sum_i \bar{g}(S_i) \, \epsilon(\bar{g}(S_i))}{\sum_i \bar{g}(S_i)} \leq \frac{\sum_i \epsilon(\bar{g}(S_i))}{n} = \bar{\epsilon}_{\mathrm{unif}},
$$
where $\bar{\epsilon}_{\mathrm{unif}}$ is the uniform (unweighted) noise rate.
\end{proposition}

\begin{proof}
By the assumption that $\epsilon(s)$ is decreasing, samples with high $\bar{g}$ (high weight) have low $\epsilon$, and vice versa. Setting $a_i = \bar{g}(S_i)$ and $b_i = \epsilon(\bar{g}(S_i))$, the monotonicity of $\epsilon$ ensures that $a_i$ and $b_i$ are oppositely ordered (i.e., whenever $a_i \leq a_j$ we have $b_i \geq b_j$), which is exactly the condition required by the Chebyshev sum inequality:
$$
n \sum_i a_i b_i \leq \left(\sum_i a_i\right)\left(\sum_i b_i\right).
$$
Substituting back and dividing both sides by $n \sum_i \bar{g}(S_i)$:
$$
\bar{\epsilon}_w = \frac{\sum_i \bar{g}(S_i)\,\epsilon(\bar{g}(S_i))}{\sum_i \bar{g}(S_i)} \leq \frac{\sum_i \epsilon(\bar{g}(S_i))}{n} = \bar{\epsilon}_{\mathrm{unif}}.
$$
\end{proof}

The inequality is strict whenever the confidence scores are non-degenerate (not all identical), meaning confidence weighting always reduces effective noise compared to uniform training on the same pseudo-labeled data. This provides theoretical grounding for the empirical observation that confidence weighting improves OOD generalization in Stage~1.

\subsection{Monotonic Improvement under Iterative Self-Training}

Prior work has established convergence guarantees for iterative self-training under various distributional assumptions, including gradual domain shift \citep{kumar2020understanding}, expansion conditions \citep{wei2020theoretical}, and mixture-model structure \citep{frei2022self}. We adapt these ideas to our two-stage curriculum and identify three sufficient conditions under which successive rounds produce non-degrading OOD performance.

\begin{proposition}[Sufficient condition for monotonic OOD improvement]
\label{prop:monotonic}
Let $f^{(t)}$ denote the detector after round $t$ of self-training, and let $R_{\mathrm{ood}}(f) = \mathbb{E}_{P_u}[\mathcal{L}(f(S), B)]$ denote the OOD risk. Assume the loss $\mathcal{L}$ is $L$-Lipschitz in the model parameters. Suppose:
\begin{enumerate}[label=(\roman*)]
    \item \textbf{Improving pseudo-labels:} The precision-coverage product satisfies $\rho_q^{(t+1)} \pi_q^{(t+1)} \geq \rho_q^{(t)} \pi_q^{(t)}$, and the resulting OOD risk reduction from Stage~1 is at least $\Delta^{(t)} > 0$;
    \item \textbf{Bounded Stage~2 drift:} Stage~2 fine-tuning runs for $T_2$ steps with learning rate $\eta_2$ and gradient norm bounded by $G$, so the parameter drift satisfies $\|f_{\mathrm{stage2}}^{(t+1)} - f_{\mathrm{stage1}}^{(t+1)}\| \leq T_2 \eta_2 G =: \delta_2$;
    \item \textbf{Bounded Stage~1 forgetting:} The gradual unfreezing schedule limits catastrophic forgetting such that $\|f_{\mathrm{stage1}}^{(t+1)} - f^{(t)}\|_\infty \leq \beta$, where $\beta > 0$ is controlled by the backbone learning rate $\gamma \eta_{\mathrm{base}}$ and the number of unfrozen blocks $k(e)$ (Remark~\ref{prop:unfreezing}).
\end{enumerate}
If $L \delta_2 \leq \Delta^{(t)}$, i.e., Stage~2 parameter drift is small enough that Lipschitz-induced risk increase does not exceed the Stage~1 gain, then $R_{\mathrm{ood}}(f^{(t+1)}) \leq R_{\mathrm{ood}}(f^{(t)})$.
\end{proposition}

\begin{proof}[Proof sketch]
By Proposition~\ref{prop:dro}, improving $\rho_q \pi_q$ (condition i) reduces $d_{\mathrm{TV}}(P_{\mathrm{eff}}^{(t+1)}, P_u)$. Since $d_{\mathrm{TV}}$ upper-bounds the $\mathcal{H}\Delta\mathcal{H}$-divergence of \citet{ben2010theory}, this tightens their generalization bound $R_u(f) \leq R_s(f) + d_{\mathcal{H}\Delta\mathcal{H}}(P_s, P_u) + \lambda$, yielding a concrete OOD risk reduction $\Delta^{(t)} > 0$ proportional to the TV distance decrease: $R_{\mathrm{ood}}(f_{\mathrm{stage1}}^{(t+1)}) \leq R_{\mathrm{ood}}(f^{(t)}) - \Delta^{(t)}$. Condition (iii) ensures that optimization remains within a $\beta$-ball of the previous iterate, so the ERM solution does not escape the basin where the tighter bound applies. For Stage~2, by Lipschitz continuity and condition (ii): $|R_{\mathrm{ood}}(f_{\mathrm{stage2}}^{(t+1)}) - R_{\mathrm{ood}}(f_{\mathrm{stage1}}^{(t+1)})| \leq L\delta_2$. Since $L\delta_2 \leq \Delta^{(t)}$, we have $R_{\mathrm{ood}}(f^{(t+1)}) \leq R_{\mathrm{ood}}(f_{\mathrm{stage1}}^{(t+1)}) + L\delta_2 \leq R_{\mathrm{ood}}(f^{(t)}) - \Delta^{(t)} + L\delta_2 \leq R_{\mathrm{ood}}(f^{(t)})$.
\end{proof}

\noindent \textbf{Empirical verification.} Our design enforces these conditions in practice: (i)~a better detector at each round produces higher-quality pseudo-labels, verified by increasing pseudo-label F1 across rounds; (ii)~Stage~2 uses a low learning rate ($\eta_2 = 5{\times}10^{-5}$, backbone at $0.02{\times}$) for only $T_2{=}40$ epochs with gradient clipping at $G{=}1.0$, bounding parameter drift $\delta_2$; gradient clipping locally enforces the Lipschitz condition by ensuring $\|\nabla \mathcal{L}\| \leq G$ at every step; additionally, OOD-based model selection provides an operational safeguard against excessive drift; (iii)~gradual unfreezing with conservative schedules in Stage~1 (warmup${}=20$, interval${}=10$) and component-specific LR multipliers (backbone at $0.005{\times}$) bound parameter drift. See Table~\ref{tab:staged_st_hparams} for the full hyperparameter specification.

\subsection{Gradual Unfreezing as Regularized Fine-Tuning}

The asymmetry between Stage~1 with conservative unfreezing and Stage~2 with aggressive unfreezing has a natural interpretation through the lens of regularization theory.

\begin{remark}[Unfreezing schedule controls effective model complexity]
\label{prop:unfreezing}
Consider a ViT encoder with $L$ transformer blocks. At epoch $e$, let $k(e) \leq L$ blocks be unfrozen with learning rate $\eta_{\mathrm{bb}} = \gamma \eta_{\mathrm{base}}$ where $\gamma \ll 1$. The number of actively trained parameters is $d_{\mathrm{eff}}(e) = d_{\mathrm{head}} + d_{\mathrm{adapter}} + k(e) \cdot d_{\mathrm{block}}$, and the backbone's effective update magnitude scales with $\gamma k(e)$. Under the conservative Stage~1 schedule where $k(e)$ increases slowly with large warmup, the model first learns task-specific head and adapter parameters on noisy pseudo-labels without disturbing pretrained representations. Under the aggressive Stage~2 schedule where $k(e)$ increases rapidly, the full model adapts to clean expert labels.
\end{remark}

This staged complexity control prevents a key failure mode: if the full backbone is unfrozen on noisy pseudo-labeled data, the pretrained representations may overfit to systematic pseudo-label errors, a form of \emph{confirmation bias}. By keeping the backbone frozen during early Stage~1 training, the detector head first calibrates its decision boundary before the backbone features are permitted to drift.

The component-specific learning rates further partition the optimization landscape: backbone blocks receive $0.005{\times}$ the base rate in Stage~1, to preserve general audio features, versus $0.05{\times}$ in Stage~2, to allow fine-grained adaptation to the expert-labeled distribution. This $10{\times}$ ratio between stages reflects the relative trustworthiness of pseudo vs.\ expert supervision.

\section{Detailed Experimental Setup}

\subsection{Data Processing}
\label{app:processing}

\noindent \textbf{Audio preprocessing.}
All raw audio is resampled to 16\,kHz mono. We segment each recording into 10.24\,s chunks with 50\% overlap using 5.12\,s hop, yielding maximum coverage while ensuring every segment fits the MAST input size. Chunks shorter than 50\% of the target duration are discarded, and those between 50--100\% are zero-padded. Before feature extraction, each chunk is mean-centered to remove DC offset.

\noindent \textbf{Spectrogram extraction.}
We compute log-mel spectrograms using Kaldi-compatible filterbanks via \texttt{torchaudio.compliance.kaldi.fbank} with the following parameters: 25\,ms Hanning window, 10\,ms frame shift, 128 mel bins, HTK-compatible mel scale, no energy feature, and no dither. This produces a $1024 \times 128$ (time $\times$ frequency) feature matrix per chunk, with frequency range $[20, 8000]$\,Hz. These settings are identical to those used by AudioMAE \citep{AudioMAE}, ensuring compatibility with the pretrained encoder. We apply dataset-wide z-score normalization during training: $(S - \mu) / (\sigma + \epsilon)$, where $\mu, \sigma$ are computed once over $\mathcal{D}_\ell \cup \mathcal{D}_u$ and reused for all training and evaluation. For BIRDeep, normalization statistics are computed separately over the combined BIRDeep + BirdSet XCM pool.

\noindent \textbf{Bounding box coordinate conversion.}
Expert annotations are provided in (Hz, seconds). We convert to spectrogram coordinates as follows. Time: $t = \lfloor \mathrm{sec} / 0.010 \rfloor$ (10\,ms frame shift). Frequency: we apply the HTK mel transform $m(\nu) = 2595 \log_{10}(1 + \nu / 700)$, distribute 128 bins linearly in mel space, and map each annotation boundary to the nearest bin center. We clip to valid bounds and enforce positive area to ensure $t^{(2)} > t^{(1)}$, $f^{(2)} > f^{(1)}$, with minimum 1 bin and 1 frame. Because the mel filterbank covers 20--8000\,Hz, annotations with frequency content entirely above 8\,kHz are effectively excluded, and those partially above 8\,kHz are clipped to the representable range. For chunked data, annotation times are adjusted relative to the chunk start; annotations spanning chunk boundaries are included in each overlapping chunk and clipped accordingly, with events shorter than 100\,ms after clipping discarded.

\noindent \textbf{Patch tokenization.}
The encoder uses non-overlapping $16 \times 16$ patches, producing a $64 \times 8$ token grid (time $\times$ frequency). We reshape tokens to a feature map $E \in \mathbb{R}^{C \times 64 \times 8}$. The detector's FPN operates at multiple pyramid levels with strides $\{s_k\}$; box coordinates are converted to each level by dividing corners by the level stride. During augmentation, boxes are transformed jointly with the spectrogram and clipped to valid bounds.

\subsection{Dataset Statistics and OOD Construction}
\label{app:dataset_stats}

\begin{table}[t]
\caption{\textbf{Dataset statistics for both domains.} Summary of labeled, unlabeled, and OOD evaluation data after preprocessing into 10.24\,s log-mel spectrogram chunks. See Table~\ref{tab:birdeep_sites} for BIRDeep site-to-habitat mapping.}
\label{tab:dataset_stats}
\centering
\scriptsize
\setlength{\tabcolsep}{3pt}
\renewcommand{\arraystretch}{1.0}
\resizebox{0.99\linewidth}{!}{
\begin{tabular}{@{}l|ccc|cccc@{}}
\toprule
& \multicolumn{3}{c|}{\textbf{Rainforest}} & \multicolumn{4}{c}{\textbf{BIRDeep}} \\
\textbf{Split} & \textbf{\# Chunks} & \textbf{\# Boxes} & \textbf{Avg./Chunk} & \textbf{\# Chunks} & \textbf{\# Boxes} & \textbf{Avg./Chunk} & \textbf{Sites} \\
\midrule
Labeled Train   & 11,726 & 74,764  & 6.38 & 3,558 & 4,362 & 1.23 & AM1,2,3,8,10,11,16 \\
Labeled Val     & 2,513  & 15,598  & 6.21 & 762   & 902   & 1.18 & same \\
Labeled Test    & 2,513  & 16,235  & 6.46 & 762   & 990   & 1.30 & same \\
\midrule
Temporal-OOD    & 286    & 3,400   & 11.89 & 671   & 1,716 & 2.56 & AM4, AM8 \\
Site-OOD        & 550    & 6,492   & 11.80 & 1,309 & 1,613 & 1.23 & AM4, AM15 \\
\midrule
Unlabeled       & 481,245 & --- & --- & 7,051 & --- & --- & All 9 \\
+ External pretrain & --- & --- & --- & 900,000 & --- & --- & --- \\
\bottomrule
\end{tabular}
}
\vspace{-4pt}
\end{table}

Table~\ref{tab:dataset_stats} summarizes the labeled, unlabeled, and OOD evaluation data for both domains after preprocessing.

\noindent \textbf{Rainforest domain.} Our recordings were collected from 2017 to 2019 from 15 sites in East Kalimantan, Indonesia. The labeled dataset contains approximately 24 hours of annotated audio (23.93\,h), while the unlabeled pool contains approximately 661.7 hours of recordings in total, corresponding to about 27.6 days of audio and 481,245 spectrogram chunks used for masked audio pretraining. The labeled set contains 16,752 spectrogram chunks and 106,597 annotated time--frequency boxes in total, with 13,170 non-empty chunks (78.6\%). Under the in-domain split, the training, validation, and test partitions contain 11,726, 2,513, and 2,513 chunks, respectively, with similar average box density across splits. For out-of-distribution evaluation, we use two annotated settings: a temporal-OOD set from sites 13A and 13B containing 286 chunks and 3,400 boxes, and an unseen-site OOD set containing 550 chunks and 6,492 boxes.

\noindent \textbf{Bird domain.} The BIRDeep dataset \citep{marquez2025birdeep} contains 641 recordings from 9 sites across 4 habitat types in Do\~{n}ana National Park, Spain. See Table~\ref{tab:birdeep_sites} for site-to-habitat mapping. Expert annotations provide time--frequency bounding boxes for 38 bird species. We process all audio identically to the rainforest domain. After chunking, we obtain 6,633 spectrogram chunks with 6,254 bounding boxes. We use 7 sites spanning all 4 habitats for a 70/15/15 train/val/test split, corresponding to 3,558/762/762 chunks. For site-OOD evaluation, we hold out AM4 (low shrubland) and AM15 (marshland), entirely unseen during training, producing 1,309 chunks with 1,613 boxes. For temporal-OOD evaluation, we hold out two recording dates, producing 671 chunks with 1,716 boxes. There is partial overlap between the two OOD sets as AM4 appears in both, but each tests a distinct shift type. For masked pretraining, we combine unlabeled BIRDeep recordings that have 7,051 chunks with BirdSet XCM \citep{rauch2024birdset} that have 900k chunks and compute normalization statistics on this combined pool.

Both domains share the same preprocessing pipeline and evaluation protocol, enabling direct cross-domain comparison. 

\begin{table}[t]
\caption{\textbf{BIRDeep recording sites and habitat types.} All 9 autonomous monitoring sites in Do\~{n}ana National Park, Spain. Sites AM4 and AM15 are held out for site-OOD evaluation; all others are used for labeled training/validation/test.}
\label{tab:birdeep_sites}
\centering
\small
\renewcommand{\arraystretch}{1.06}
\setlength{\tabcolsep}{5pt}
\begin{tabular}{@{}llll@{}}
\toprule
\textbf{Recorder} & \textbf{Place Name} & \textbf{Habitat} & \textbf{Split} \\
\midrule
AM1  & Monteblanco      & Low shrubland  & Train \\
AM2  & Sabinar          & High shrubland & Train \\
AM3  & Ojillo           & High shrubland & Train \\
AM4  & Pozo Sta Olalla  & Low shrubland  & Site-OOD \\
AM8  & Torre Palacio    & Ecotone        & Train \\
AM10 & Pajarera         & Ecotone        & Train \\
AM11 & Ca\~{n}o Martinazo & Ecotone      & Train \\
AM15 & Cancela Mill\'{a}n & Marshland    & Site-OOD \\
AM16 & Juncabalejo      & Marshland      & Train \\
\bottomrule
\end{tabular}
\vspace{-4pt}
\end{table}

\subsection{Adapter Designs.}
\label{app:adapter}

We study three adapter variants that bridge the pretrained ViT-B encoder ($C{=}768$, spatial $64{\times}8$) to the FPN-based detection head.

\begin{enumerate}[label=(\roman*)]
\item \textbf{Conv adapter:} Three $3{\times}3$ conv layers ($768{\to}512{\to}256{\to}256$) with ReLU, preserving spatial resolution at $64{\times}8$. Restores short-range locality lost by global ViT attention but cannot resolve sub-patch events.

\item \textbf{Anisotropic adapter:} Stem conv ($768{\to}512$), residual block, separable $3{\times}1$ / $1{\times}3$ convolutions (Proposition~\ref{prop:separable}), channel reduction ($512{\to}256$), second residual block, and nearest-neighbor frequency upsampling ($8{\to}32$). Operates at $64{\times}32$.

\item \textbf{Upsampling adapter} (default): $1{\times}1$ channel reduction ($768{\to}256$), residual block at base resolution, learnable transposed-conv upsampling (${\times}2$ time, ${\times}4$ frequency), separable $3{\times}1$ / $1{\times}3$ anisotropic refinement at high resolution (Proposition~\ref{prop:separable}), and a final $3{\times}3$ residual refinement block. Operates at $128{\times}32$ with effective strides $8{\times}4$ (Proposition~\ref{prop:optimal_upsample}).
\end{enumerate}

All adapters output 256 channels and add 2--4\,M parameters, keeping $>$85\% of capacity in the frozen pretrained encoder.

\subsection{Architecture Details}
\label{app:architecture}

\noindent \textbf{Neck and fusion.} Starting from the adapter output (P3), we build P4/P5 via stride-2 $3 \times 3$ convs. We use lateral $1 \times 1$ projections and top-down fusion with non-negative learnable weights normalized by their sum, and apply lightweight smoothing blocks (mixing $3 \times 3,3 \times 1,1 \times 3$ ) to preserve narrow-band detail while expanding temporal context.

\noindent \textbf{Detector heads.} We use FCOS, predicting per-location classification, centerness, and box regression with center sampling on the spectrogram lattice. We also use ATSS assignment \citep{zhang2020atss}, Quality Focal Loss (QFL) \citep{li2020qfl} for classification, and GIoU \citep{rezatofighi2019giou} for regression.

\subsection{Optimization details.}
\label{app:hyperparam}


For masked-audio pretraining, we use AdamW with mixed-precision training and gradient clipping. The pretrained MAST encoder is initialized from the checkpoint pretrained on AudioSet \citep{gemmeke2017audio} and optimized on unlabeled rainforest spectrograms using a masked reconstruction objective. Table~\ref{tab:pretrain_hparams} summarizes the main pretraining hyperparameters.

\begin{table}[t]
\caption{\textbf{Masked Audio pretraining hyperparameters.} Main optimization settings used for masked-audio pretraining on unlabeled rainforest mel-spectrograms.}
\label{tab:pretrain_hparams}
\centering
\small
\renewcommand{\arraystretch}{1.06}
\setlength{\tabcolsep}{6pt}
\begin{tabular}{@{}ll@{}}
\toprule
\textbf{Hyperparameter} & \textbf{Value} \\
\midrule
Model & \texttt{mae\_vit\_base\_patch16} \\
Optimizer & AdamW \\
AdamW betas & $(0.9,\,0.95)$ \\
Batch size & 64 \\
Training epochs & 400 \\
Base learning rate & $5\times10^{-5}$ \\
Minimum learning rate & $5\times10^{-7}$ \\
Warmup epochs & 20 \\
Weight decay & 0.05 \\
Mask ratio & 0.75 \\
Time masking probability & 0.6 \\
Frequency masking probability & 0.3 \\
Gradient clipping & 1.0 \\
Mixed precision & AMP \\
\bottomrule
\end{tabular}
\vspace{-4pt}
\end{table}

For MAST detector fine-tuning, we also use AdamW together with a warmup--cosine learning-rate schedule. The detector is initialized from the best masked-audio pretraining checkpoint and trained for 200 epochs on the labeled spectrogram data. For MAST detectors, we use a smaller effective learning rate for the pretrained backbone and larger learning rates for newly introduced detector components. We train with mixed moderate spectrogram augmentation, exponential moving average, and gradient clipping. Table~\ref{tab:detector_hparams} summarizes the main fine-tuning hyperparameters used for the MAST model.

\begin{table}[t]
\caption{\textbf{Detector fine-tuning hyperparameters.} Main optimization settings used for MAST detector training.}
\label{tab:detector_hparams}
\centering
\small
\renewcommand{\arraystretch}{1.06}
\setlength{\tabcolsep}{6pt}
\begin{tabular}{@{}ll@{}}
\toprule
\textbf{Hyperparameter} & \textbf{Value} \\
\midrule
Backbone initialization & Best masked-audio pretraining checkpoint \\
Adapter type & Upsampling \\
Neck & FPN P2-P5 \\
Optimizer & AdamW \\
Batch size & 64 \\
Training epochs & 200 \\
Base learning rate & $3\times10^{-4}$ \\
Backbone LR multiplier & 0.1 \\
Head LR multiplier & 1.0 \\
Weight decay & $1\times10^{-3}$ \\
Gradient clipping & 1.0 \\
Learning-rate schedule & Warmup + cosine decay \\
EMA decay & 0.9995 \\
Augmentation probability & 0.9 \\
Contrastive weight & 0.5 \\
Contrastive temperature & 0.1 \\
Contrastive feature dimension & 256 \\
Contrastive projection dimension & 128 \\
Contrastive negatives per image & 32 \\
\bottomrule
\end{tabular}
\vspace{-4pt}
\end{table}

For staged self-training, we train the model in two phases with different optimization regimes. Stage~1 starts from the best supervised checkpoint and trains on large pseudo-labeled unlabeled data using confidence weighting, moderate augmentation, EMA, contrastive learning, and conservative gradual unfreezing. Its goal is to expand coverage under distribution shift while limiting drift from noisy pseudo supervision. Stage~2 then initializes from the best Stage~1 checkpoint and fine-tunes on clean expert-labeled data with a lower base learning rate, slightly weaker augmentation, and a more aggressive gradual-unfreezing schedule. This second phase re-anchors the detector to high-quality annotations while preserving the broader domain exposure acquired in Stage~1. In both stages, we use AdamW, warmup--cosine scheduling, and gradient clipping. Table~\ref{tab:staged_st_hparams} summarizes the main hyperparameters.

\noindent \textbf{Compute resources.} All experiments were conducted on a single NVIDIA A100 GPU (80\,GB).

\begin{table}[t]
\caption{\textbf{Staged self-training hyperparameters.} Main optimization settings for Stage~1 pseudo-labeled exploration and Stage~2 expert-labeled refinement. }
\label{tab:staged_st_hparams}
\centering
\small
\renewcommand{\arraystretch}{1.06}
\setlength{\tabcolsep}{6pt}
\begin{tabular}{@{}lll@{}}
\toprule
\textbf{Hyperparameter} & \textbf{Stage 1} & \textbf{Stage 2} \\
\midrule
Training data & Pseudo-labeled unlabeled data & Expert-labeled data \\
Initialization & Best supervised checkpoint & Best Stage~1 checkpoint \\
Adapter / neck & Upsampling + FPN P2-P5 & Upsampling + FPN P2-P5\\
Optimizer & AdamW & AdamW \\
Batch size & 64 & 64 \\
Training epochs & 80 & 40 \\
Base learning rate & $1\times10^{-4}$ & $5\times10^{-5}$ \\
Weight decay & $1\times10^{-3}$ & $1\times10^{-3}$  \\
Gradient clipping & 1.0 & 1.0 \\
Scheduler & Warmup + cosine decay & Warmup + cosine decay \\
Warmup epochs & 8 & 4 \\
Backbone LR multiplier & 0.1 & 0.02 \\
Adapter LR multiplier & 1.0 & 1.0 \\
Head LR multiplier & 1.0 & 1.0 \\
Unfreezing warmup & 20 epochs & 3 epochs \\
Unfreezing interval & Every 10 epochs & Every 6 epochs \\
Blocks per stage & 3 & 6 \\
Unfrozen backbone LR multiplier & 0.005 & 0.05 \\
Confidence weighting & Enabled & Not Enabled \\
Pseudo-label confidence threshold & 0.10 & -- \\
Data augmentation prob. & 0.9 & 0.8 \\
EMA decay & 0.9995 & 0.9995 \\
Contrastive loss weight & 0.3 &  0.2 \\
Contrastive temperature & 0.1 & 0.1 \\
Contrastive feature / proj. dim & 256 / 128 & 256 / 128 \\
\bottomrule
\end{tabular}
\vspace{-4pt}
\end{table}

\noindent \textbf{Pseudo-label confidence threshold.}
The threshold $q=0.10$ reflects the low-confidence regime typical of bioacoustic detectors, where true positive scores concentrate well below those seen in natural image detection due to faint, distant, or overlapping animal calls. At this operating point the threshold retains the majority of true positives while filtering out the lowest-scoring false alarms. Remaining label noise is further suppressed by confidence weighting, which down-weights uncertain pseudo-labels in proportion to their average detection score (Section~\ref{sec:method}, Proposition~\ref{prop:confidence_weighting}).

\subsection{Augmentations}
\label{app:data_augmentations}

During supervised detector training, we apply moderate spectrogram augmentations with box-consistent transformations. All augmentations operate on $128 \times 1024$ log-mel spectrogram chunks and transform the associated time--frequency boxes jointly with the input. After each geometric augmentation, boxes are clipped to valid spectrogram boundaries and filtered if they become too small to remain meaningful training targets.

Our augmentation pipeline includes the following components:

\begin{itemize}[leftmargin=14pt]
    \item \textbf{Time shift.} We randomly shift the spectrogram along the time axis with zero/min-value padding, and shift all box time coordinates accordingly.
    \item \textbf{Frequency shift.} We randomly shift the spectrogram along the mel-frequency axis and apply the same offset to the box frequency coordinates.
    \item \textbf{Time warp.} We apply mild temporal warping by rescaling the spectrogram along the time axis and mapping box boundaries through the same transformation.
    \item \textbf{Frequency masking.} We mask a random frequency band in the style of SpecAugment \citep{park2019specaugment}. Boxes with heavy overlap with the masked band are removed.
    \item \textbf{Time masking.} We mask a random temporal segment and similarly discard boxes that are largely occluded by the mask.
    \item \textbf{Gaussian noise and volume scaling.} We apply small additive Gaussian noise and moderate global intensity scaling to improve robustness to recording and gain variation.
    \item \textbf{Background mixing.} We optionally mix a spectrogram with background-noise samples drawn from a curated noise bank from our training data containing common environmental and anthropogenic sounds such as rain, thunder, aircraft, chainsaw, and vehicle noise, which improves robustness to realistic acoustic interference.
    \item \textbf{Mixup and CutMix.} We additionally implemented Mixup and CutMix for spectrograms by combining two training examples and merging or clipping their boxes accordingly. These augmentations are designed to simulate overlapping events and partial occlusion.
\end{itemize}

To avoid destabilizing detection, we use only mild transformation strengths and sample a small random subset of augmentations for each training example. This design preserves biologically meaningful event structure while improving robustness to temporal shifts, frequency variation, and background-condition changes.

\subsection{Evaluation Details}
\label{app:eval}

We evaluate all models on binary time--frequency detection of ``any animal sound'' using five metrics: precision, recall, F1 score, mean average precision (mAP), and mean intersection-over-union (mIoU). A prediction is counted as a true positive if it overlaps a ground-truth box with IoU $\geq 0.5$. Remaining predictions are false positives and unmatched ground-truth boxes are false negatives. We apply class-agnostic NMS to detector outputs with fixed overlap thresholds and confidence thresholds tuned on the in-domain validation set, and report metrics aggregated over all chunks in each test split. OOD split construction is described in Sec.~\ref{app:dataset_stats}. Unless otherwise stated, all numbers collapse all sonotypes into a single foreground class.

\section{Additional Experimental Results}

\begin{table}[tb!]
\caption{\textbf{Adapter design ablation on cross-site OOD detection on rainforest domain.}
All variants use the same MAST detector and training setup; only the adapter module is changed. \textit{Higher is better}.}
\label{tab:adapter_ablation_ood}
\centering
\small
\renewcommand{\arraystretch}{1.08}
\resizebox{0.99\linewidth}{!}{
\begin{tabular}{@{}l|ccccc@{}}
\toprule
\textbf{Adapter Variant} & \textbf{Precision} & \textbf{Recall} & \textbf{F1} & \textbf{mAP} & \textbf{Mean IoU} \\
\midrule
Conv adapter & 0.5928 & 0.3517 & 0.4415 & \textbf{0.3004} & 0.6644 \\
Anisotropic adapter & 0.6432 & \textbf{0.3535} & \textbf{0.4563} & 0.2866 & 0.6466 \\
Upsampling adapter (Main) & \textbf{0.7562} & 0.3169 & 0.4466 & 0.2573 & \textbf{0.6824} \\
\bottomrule
\end{tabular}
}
\vspace{-4pt}
\end{table}

\begin{figure*}[!t]
  \centering
  \begin{subfigure}[t]{0.32\linewidth}
    \centering
    \includegraphics[width=\linewidth]{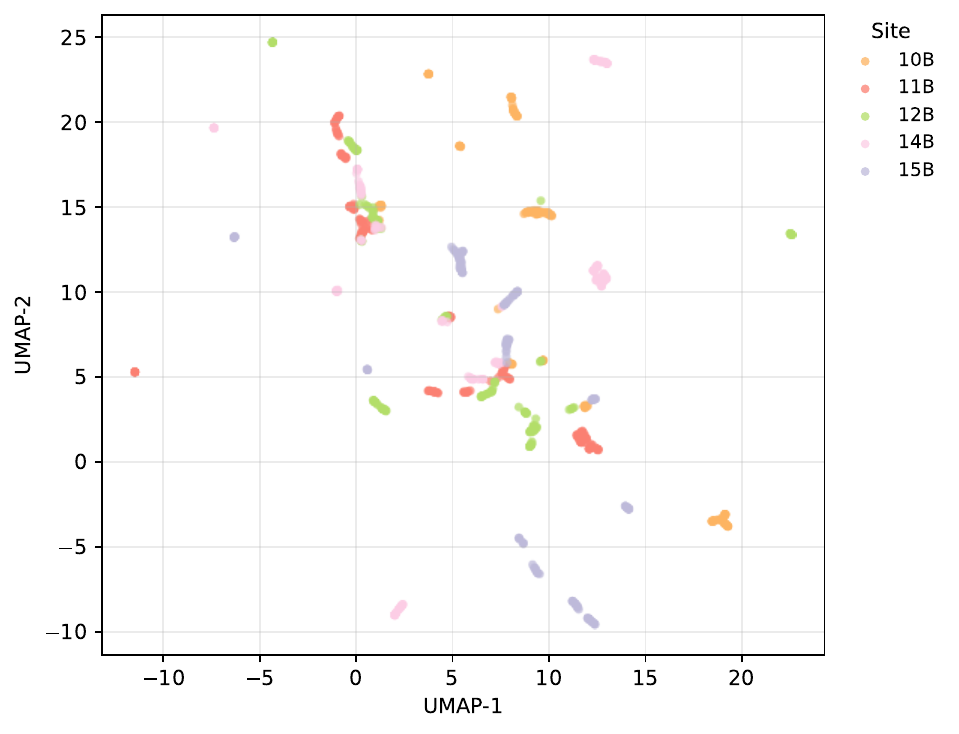}
    \caption{Colored by recording site.}
    \label{fig:umap_site}
  \end{subfigure}
  \hfill
  \begin{subfigure}[t]{0.32\linewidth}
    \centering
    \includegraphics[width=\linewidth]{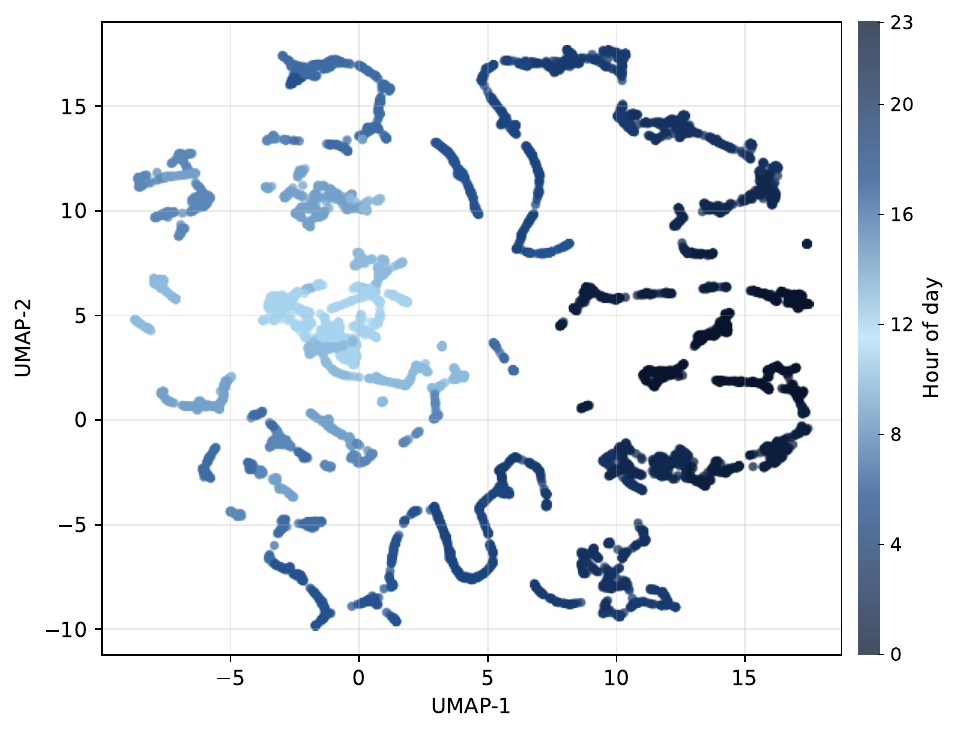}
    \caption{Colored by hour of day.}
    \label{fig:umap_time}
  \end{subfigure}
  \hfill
  \begin{subfigure}[t]{0.32\linewidth}
    \centering
    \includegraphics[width=\linewidth]{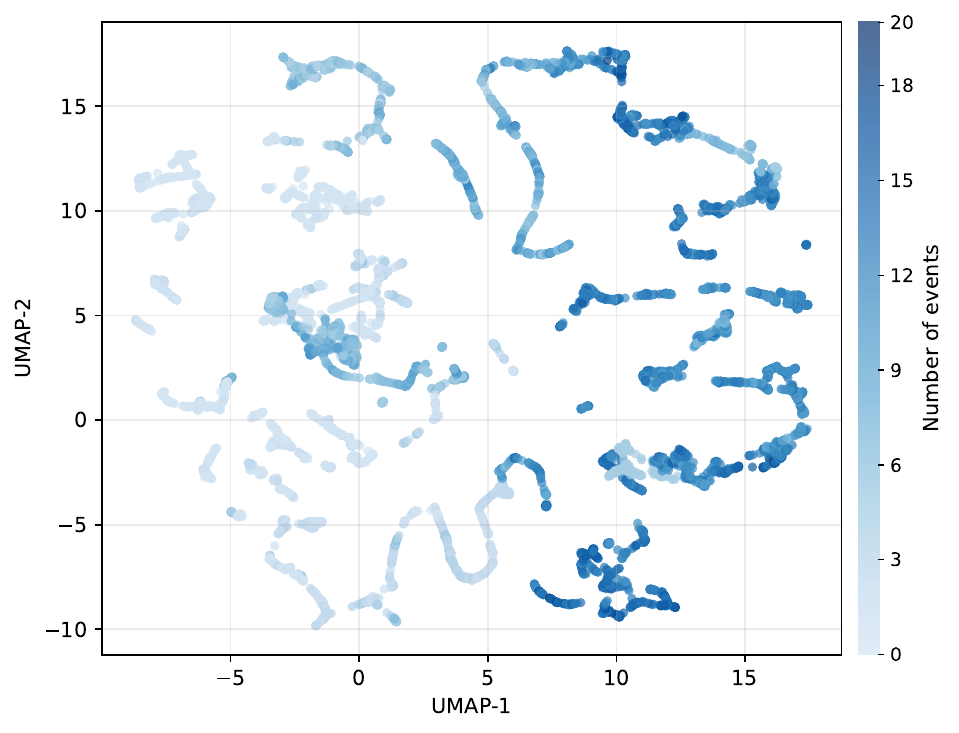}
    \caption{Colored by annotated event count.}
    \label{fig:umap_eventcount}
  \end{subfigure}
  \caption{\textbf{UMAP visualization of pretrained encoder embeddings before detector fine-tuning on rainforest domain.} Each point is one mel-spectrogram chunk projected from the pretrained representation space into two dimensions. (a)~Coloring by recording site reveals spatial clustering, indicating that the encoder captures site-specific acoustic signatures. (b)~Coloring by hour of day shows a diel gradient, reflecting systematic changes in background noise and species activity across the 24-hour cycle. (c)~Coloring by annotated event count shows that chunks with more animal-sound activity occupy partially distinct regions from background-dominated chunks, suggesting that masked-audio pretraining captures biologically relevant variation before any supervised training. Together, these patterns provide qualitative evidence of both the usefulness of domain-matched pretraining and the presence of temporal and cross-site distribution shift.}
  \label{fig:umap_pretrained}
\end{figure*}

\subsection{Per-Round Self-Training Results on BIRDeep}
\label{app:birdeep_st_rounds}

Table~\ref{tab:birdeep_st_rounds} reports per-round self-training results on BIRDeep across in-domain, temporal-OOD, and cross-site-OOD evaluation. Self-training produces large mAP gains across all settings: in-domain mAP rises from 0.601 at Round~0 to 0.778 at Round~2, a +0.18 absolute improvement that surpasses all baselines including FCOS at 0.665. Under cross-site shift, mAP improves from 0.380 to 0.441, and under temporal shift from 0.384 to 0.428. F1 also improves consistently across OOD settings, reaching 0.643 cross-site and 0.649 temporal at Round~2. These trends mirror the monotonic OOD gains observed on the rainforest domain (Figure~\ref{fig:self-training}), confirming that the two-stage self-training curriculum generalizes across ecologically distinct domains.

\begin{table}[t]
\caption{\textbf{Per-round self-training results on BIRDeep.} Round~0 is the supervised MAST detector before self-training. \textit{Higher is better}.}
\label{tab:birdeep_st_rounds}
\centering
\small
\renewcommand{\arraystretch}{1.05}
\setlength{\tabcolsep}{5pt}
\begin{tabular}{@{}llccccc@{}}
\toprule
\textbf{Setting} & \textbf{Round} & \textbf{Prec.} & \textbf{Rec.} & \textbf{F1} & \textbf{mAP} & \textbf{mIoU} \\
\midrule
& Round 0 (no ST) & \textbf{0.7824} & 0.7010 & 0.7395 & 0.6011 & 0.6822 \\
In-domain & Round 1 & 0.7054 & 0.8248 & 0.7604 & 0.6870 & 0.7080 \\
& Round 2 & 0.6709 & \textbf{0.9061} & \textbf{0.7709} & \textbf{0.7783} & \textbf{0.7339} \\
\midrule
& Round 0 (no ST) & \textbf{0.7618} & 0.4808 & 0.5895 & 0.3838 & \textbf{0.6083} \\
Temporal OOD & Round 1 & 0.7253 & 0.5780 & 0.6434 & 0.3840 & 0.6020 \\
& Round 2 & 0.6755 & \textbf{0.6247} & \textbf{0.6491} & \textbf{0.4280} & 0.5903 \\
\midrule
& Round 0 (no ST) & 0.6476 & 0.5059 & 0.5680 & 0.3803 & \textbf{0.6001} \\
Cross-site OOD & Round 1 & \textbf{0.6893} & 0.5914 & 0.6367 & 0.3953 & 0.5963 \\
& Round 2 & 0.6522 & \textbf{0.6438} & \textbf{0.6434} & \textbf{0.4414} & 0.5943 \\
\bottomrule
\end{tabular}
\vspace{-4pt}
\end{table}

\subsection{Additional Adapter Design Ablation.}
\label{app:more_ablation}
Table \ref{tab:adapter_ablation_ood} compares adapter designs while holding the backbone, training recipe, and detection head fixed. All three adapters enable strong OOD performance, but they emphasize different operating points. The anisotropic adapter achieves the best F1 by slightly improving recall while keeping precision competitive, which aligns with the common structure of animal calls being temporally extended and frequency-concentrated. The convolutional adapter achieves the strongest mAP, suggesting that a simple locality-restoring module can already improve ranking across diverse OOD backgrounds. We use the upsampling adapter in our main result because it yields the highest precision and mean IoU, producing more confident and spatially accurate detections, properties that directly benefit downstream self-training, where high-precision pseudo-labels reduce confirmation bias. Overall, these results show that adapter choice primarily controls the precision--recall and box-quality trade-off, while the combination of in-domain pretraining and the detection stack remains the dominant driver of OOD robustness.

\begin{figure}[t]
\centering
\includegraphics[width=0.78\linewidth]{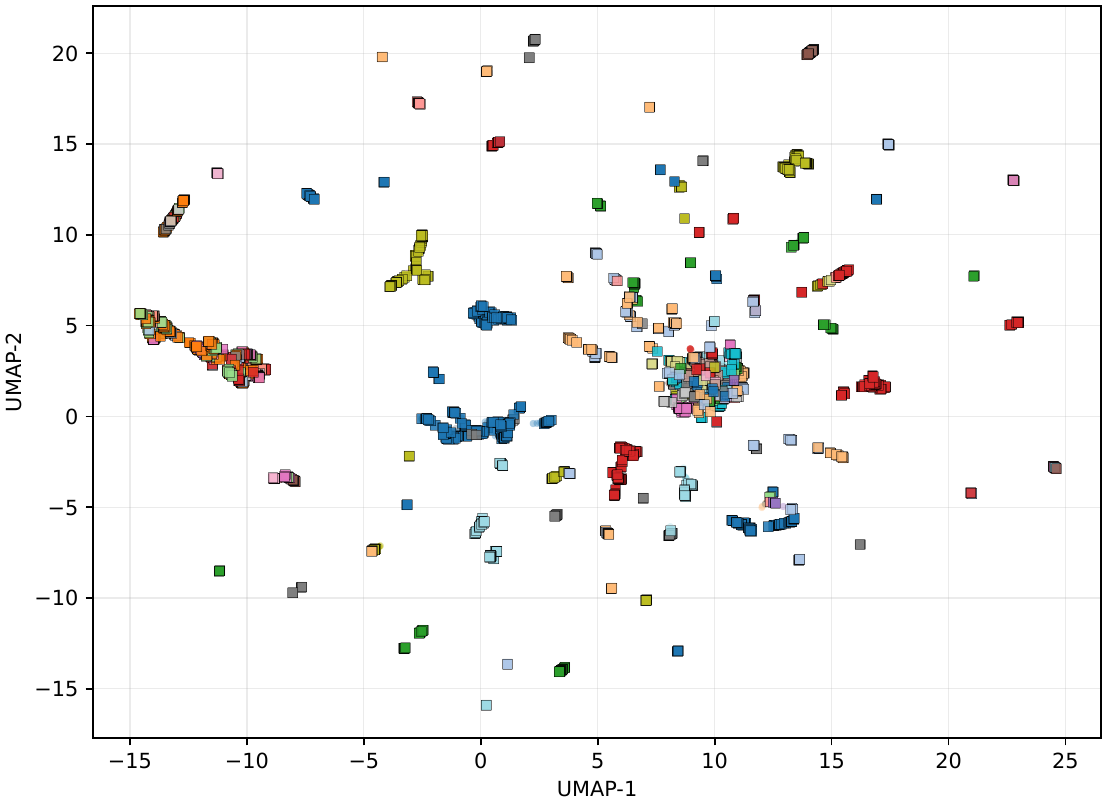}
\caption{\textbf{Box-level embedding space induced by MAST detections on rainforest domain.} Each point corresponds to a ground-truth or predicted time--frequency box, embedded by pooling its corresponding region from the pretrained encoder feature map and projected with UMAP. Boxes with similar acoustic structure form partially coherent clusters, and high-quality predictions tend to lie near ground-truth embeddings of the same or related sonotypes. This suggests a practical downstream use of MAST: detected sound events can be retrieved against annotated examples for tentative zero-shot classification, while acoustically distinct clusters can be surfaced to experts as candidate novel sounds. Colors denote sonotypes; circles correspond to predicted boxes and squares correspond to ground-truth boxes.}
\label{fig:box_embedding_application}
\end{figure}

\subsection{Representation Analysis of the Pretrained Encoder}
\label{app:representations}

To better understand what the pretrained encoder captures before detector fine-tuning, we visualize spectrogram-level embeddings using UMAP \citep{mcinnes2020umapuniformmanifoldapproximation}. For each spectrogram chunk, we extract the pretrained encoder representation, project it into two dimensions, and color the same embedding space by three attributes: recording site, hour of day, and annotated event count.

Figure~\ref{fig:umap_pretrained} provides two qualitative insights. First, the visible gradient with respect to annotated event count suggests that masked-audio pretraining captures variation related to animal-sound activity, even before detector fine-tuning. Second, the organization by site and hour of day indicates that the learned representation also reflects domain-dependent structure associated with habitat, diel cycle, and background acoustics. These patterns qualitatively support both the usefulness of masked-audio pretraining for downstream detection and the presence of temporal and cross-site distribution shift in rainforest soundscapes. We emphasize that UMAP is only a qualitative visualization tool rather than a quantitative measure of separability, but it offers an intuitive view of why both strong pretraining and domain shift aware adaptation are important in this setting.

\begin{table}[!t]
\caption{\textbf{Retrieval-based sonotype classification from box embeddings on rainforest domain, with 131 sonotypes.} We evaluate nearest-neighbor retrieval against annotated training boxes in the learned embedding space. \emph{Oracle-box classification} uses only ground-truth boxes, while \emph{detected-box retrieval} uses predicted boxes matched to ground truth at IoU $\geq 0.5$. Higher is better for Top-1, Top-5; lower is better for mean rank.}
\label{tab:retrieval_classification}
\centering
\small
\renewcommand{\arraystretch}{1.06}
\setlength{\tabcolsep}{6pt}
\begin{tabular}{@{}lccccc@{}}
\toprule
\textbf{Setting} & \textbf{\# Queries} & \textbf{\# Labels} & \textbf{Top-1} & \textbf{Top-5} & \textbf{Mean Rank} \\
\midrule
Oracle-box classification (GT boxes) & 14,986 & 131 & 0.723 & 0.902 & 5.76 \\
Detected-box retrieval (matched preds) & 8,414 & 104 & 0.961 & 0.982 & 3.22 \\
\bottomrule
\end{tabular}
\vspace{-4pt}
\end{table}

\subsection{Application: Box-Level Embeddings for Retrieval and Zero-Shot Classification}
\label{app:box_embedding_application}

An additional advantage of MAST box-level localization is that each detected sound event can be embedded and compared directly in a shared representation space. After detection, we extract box-level features from the pretrained encoder by pooling the corresponding time--frequency region from the feature map, yielding one embedding per detected event. This enables a retrieval-style downstream workflow in which detected boxes can be matched to annotated training boxes, grouped by acoustic similarity, or provide detections to experts as candidate novel sounds.

Figure~\ref{fig:box_embedding_application} illustrates this application qualitatively. Ground-truth and predicted boxes form partially structured clusters in the learned embedding space, and predictions with high overlap to ground truth tend to lie near embeddings of the same or related sonotypes. To quantify this effect, we also evaluate retrieval-based sonotype classification in two settings: (i) \emph{oracle-box classification}, which uses only ground-truth boxes to test whether the embedding space discriminates sonotypes, and (ii) \emph{detected-box retrieval}, which uses predicted boxes matched to ground truth at IoU $\geq 0.5$ and assigns each matched prediction the sonotype of its best-matching ground-truth box. In both cases, we retrieve nearest neighbors from the annotated training set and report Top-1, Top-5, and mean rank.

As shown in Table~\ref{tab:retrieval_classification}, the learned box embeddings support strong retrieval-based classification. Oracle-box classification achieves 72.3\% Top-1 and 90.2\% Top-5 accuracy over 131 sonotypes, indicating that the embedding space already captures substantial sonotype structure. Retrieval on matched predicted boxes is even stronger with 96.1\% Top-1 and 98.2\% Top-5, reflecting that high-quality detections tend to align with acoustically prototypical events. In hyperdiverse and acoustically understudied environments, such a workflow is practically useful: detected sounds can be assigned tentative labels through retrieval against known examples, while acoustically distinct clusters can be prioritized for expert review as potentially novel or rare events.

\subsection{Label Budget and Sampling Strategy}
\label{app:label_budget}
\label{app:data_scaling}

\begin{figure}[t]
\centering
\includegraphics[width=0.78\linewidth]{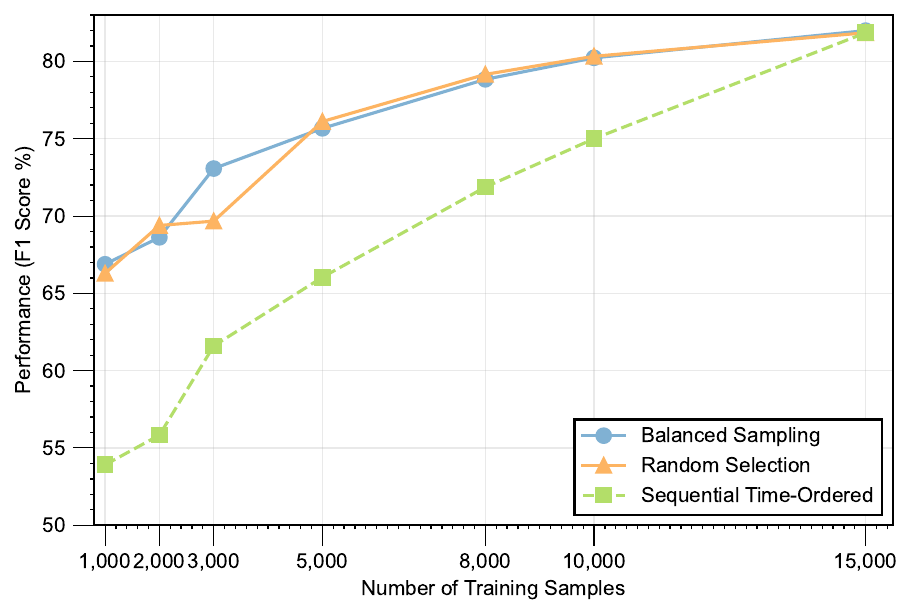}
\caption{\textbf{Label-efficiency study under different sampling strategies on rainforest domain with FCOS}. We vary the number of labeled training chunks (x-axis) and report F1 on the held-out test set (y-axis). Time-balanced sampling across the 24-hour diel cycle consistently outperforms sequential labeling at low budgets, and remains slightly better than random sampling for most budgets. As the label budget increases, the gap narrows and all strategies converge to F1 $\approx$ 0.82 at $N{=}15000$.}
\label{fig:data-scaling}
\end{figure}

We study how label allocation affects performance under a fixed annotation budget. Starting from the full labeled training split, we subsample $N$ chunks for $N \in \{250, 500, 1000, \ldots\}$ to simulate different labeling budgets. For each $N$, we compare three sampling strategies:

\begin{itemize}
\item \textbf{Random.} Uniformly sample $N$ chunks from the training split without replacement.
\item \textbf{Time-balanced.} Group chunks by their hour-of-day and sample approximately $N/24$ chunks per hour, ensuring coverage of both day and night.
\item \textbf{Sequential.} Take the first $N$ chunks in temporal order, mimicking a naive labeling strategy where annotators start at the beginning and stop once a budget is exhausted.
\end{itemize}

For each pair we train a FCOS detector from scratch with identical optimization and augmentation settings and evaluate on the same in-domain splits. Figure~\ref{fig:data-scaling} reports F1 across budgets. At low budgets, the sampling policy is a first-order factor: sequential labeling consistently underperforms, while time-balanced sampling yields the strongest results, achieving F1 of 0.669 vs.\ 0.539, a +0.130 F1 gain at $N=1000$ . This pattern highlights the impact of temporal distribution shift even within a single site at a single day: labels drawn from one contiguous time range provide limited coverage of acoustic conditions, whereas time-balanced sampling improves robustness by exposing the model to broader diel variation. As the labeling budget increases, the gap between strategies narrows. Practically, this suggests that when annotation budgets are limited, it is better to spread labeling effort across times of day than to annotate one continuous stretch of audio. This finding is especially relevant because sequential labeling is the most natural default in practice: annotators typically open a recording and label forward in chronological order until the budget is exhausted, inadvertently concentrating all labels within a narrow temporal window.

\subsection{Qualitative Detection Comparison}
\label{app:qualitative}

To provide visual intuition for the quantitative gaps reported in Table~\ref{tab:ood_combined}, we randomly sample two 10.24\,s chunks from the cross-site OOD test set and visualize the predictions of every method side by side as in Figure~\ref{fig:qualitative}.

Several patterns are apparent. Faster R-CNN fails almost entirely under cross-site shift, producing few or no predictions and yielding near-zero recall. FCOS and DETR exhibit the opposite failure mode: both generate numerous spurious predictions scattered across the spectrogram, resulting in very low precision despite moderate recall, suggesting overfitting to the spectral statistics of the labeled site. AudioMAE shows partial improvement but still misses many events. In contrast, MAST predictions closely track the ground-truth boxes with markedly fewer false positives, consistent with its substantially higher OOD F1 in Table~\ref{tab:ood_combined}.

\begin{figure}[t]
\centering
\includegraphics[width=\linewidth]{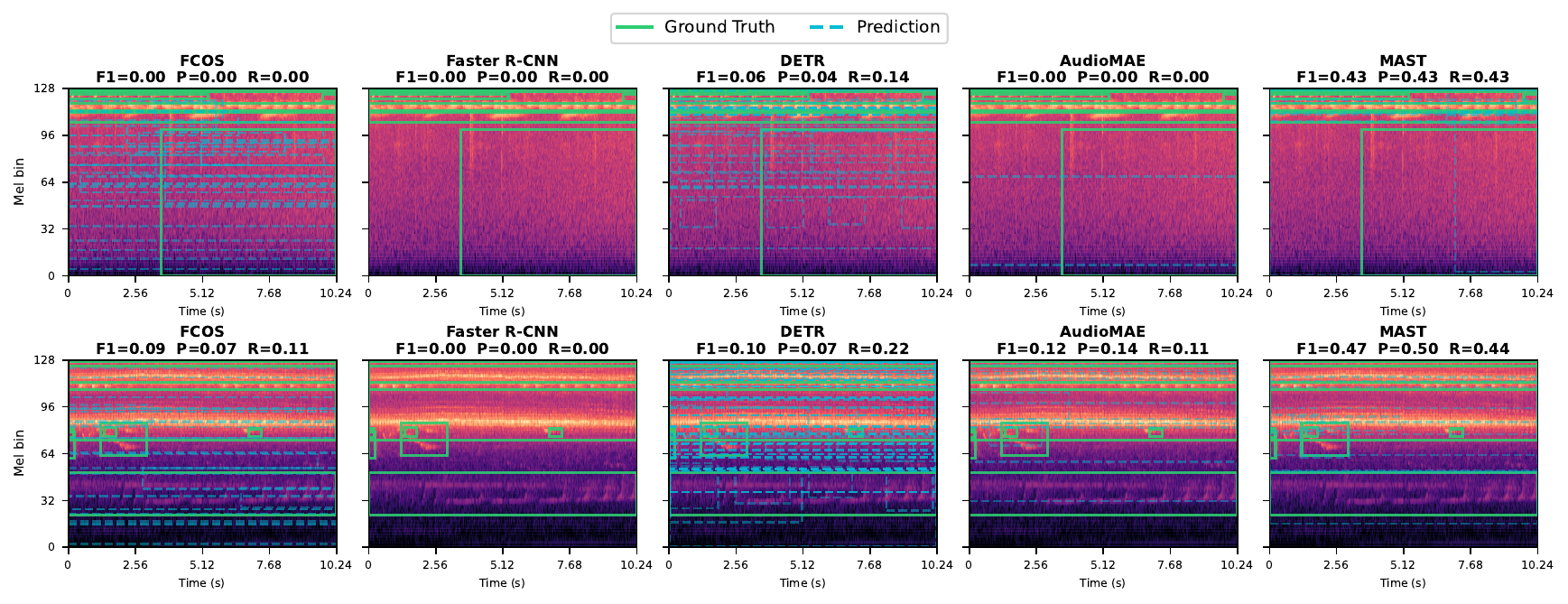}
\caption{Qualitative detection comparison on two randomly sampled cross-site OOD chunks on rainforest domain. Ground-truth boxes are shown in green solid lines; predicted boxes in blue dashed lines. Baseline detectors produce numerous false positives or miss most events entirely, while MAST localizes the majority of sound events with fewer spurious detections.}
\label{fig:qualitative}
\end{figure}


\end{document}